\pdfoutput=1
\documentclass[10pt]{article}
\usepackage{amsmath}
\usepackage{amsfonts}
\usepackage{amssymb}
\usepackage{mathrsfs}
\usepackage{amstext}
\usepackage{latexsym}
\usepackage{mathtools}
\usepackage{cite}
\usepackage{amsthm}
\usepackage{hyperref}
\hypersetup{hidelinks}
\usepackage{enumitem}
\usepackage[utf8]{inputenc}
\usepackage{authblk}
\usepackage{textcomp, gensymb}
\usepackage{ragged2e}
\usepackage{xurl}
\justifying
\theoremstyle{definition}
\newtheorem{definition}{Definition}[section] 
\newtheorem{theorem}{Theorem}[section]
\newtheorem{proposition}{Proposition}[section]
\newtheorem{remark}{Remark}[section]
\newtheorem{lemma}{Lemma}[section]
\newtheorem{corollary}{Corollary}[section]
\allowdisplaybreaks
\newcommand{\R}{\mathbb{R}}

\catcode `\@=11 \@addtoreset{equation}{section}

\catcode `\@=12

\begin{document}
\date{}
\title{Affine Left--Right Dirac Data on the Moyal Plane, with an Application to
\texorpdfstring{\(U(1)_\star\)}{U(1)*} Gauge Coupling in Noncommutative Quantum Mechanics}

\author[1]{Tanmoy Kumar Sarkar\thanks{tanmoy.sarkar@g.bracu.ac.bd}}
\author[1,2]{Md. Rafsanjany Jim\thanks{mjim3@unm.edu}}
\author[1]{S. Hasibul Hassan Chowdhury\thanks{shhchowdhury@bracu.ac.bd}}

\affil[1]{Department of Mathematics and Physical Sciences, BRAC University, Kha 224 Bir Uttam Rafiqul Islam Avenue, Merul Badda, Dhaka, Bangladesh}
\affil[2]{Department of Mathematics and Statistics, University of New Mexico, Albuquerque, NM, USA}

\maketitle

\begin{abstract}

We study even non-unital Dirac data on the two-dimensional Moyal plane with fixed deformation parameter \(\vartheta\neq0\), whose Dirac operator is affine in the left and right regular coordinate multipliers, with coefficient matrices \(T_L\) and \(T_R\). Two invariants organize this class. The right block controls algebra--Dirac transversality: local compactness holds if and only if \(\det T_R\neq0\), with \(|\det T_R|\) determining the local volume coefficient and \(\operatorname{sgn}\det T_R\) the oriented Gaussian \(K\)-homology pairing. Independently, \(\Gamma=\vartheta(\det T_L-\det T_R)\) controls the spectral regime and, together with \(|\det T_R|\), the exact localized heat trace: \(\Gamma\neq0\) gives infinitely degenerate Landau levels, while for \(\Gamma=0\) in the transverse regime the spectrum is \(\mathbb R\) with empty point spectrum. Applied to Wallet's Landau--Moyal family, the criterion shows that local compactness survives at the metric-degenerate point \(\xi_{\mathrm W}=1\) but fails at \(\xi_{\mathrm W}=-1\), thereby resolving the local-compactness behaviour at both exceptional one-sided points. We then derive a distinguished affine family from a fixed nondegenerate sector \((\hbar_0,\vartheta_0,B_0)\) of the kinematical symmetry group of two-dimensional noncommutative quantum mechanics. The sector fixes the effective Moyal scale \(\vartheta_{\mathrm{eff}}=\vartheta_0/(1-\vartheta_0B_0/\hbar_0)\) and \(T_R\), while an affine \(U(1)_\star\) background coupling deforms only the left affine part of the Dirac operator. Consequently, local compactness, noncommutative volume and \(K\)-homology class remain unchanged, whereas \(\Gamma=\tfrac12(1-eF_{xy}/B_0)\) drives a transition at \(eF_{xy}=B_0\) from Landau levels to spectrum \(\mathbb R\) with empty point spectrum.
\end{abstract}

{

}

\section{Introduction}\label{sec:intro}

Non-unital spectral geometry on Moyal spaces is well established, beginning
with the standard Moyal spectral triple \cite{gayral2004moyal} and including
harmonic and Landau-type variants
\cite{GayralWulkenhaar2013,wallet2012homothetic}. Spectral-triple
constructions also occur directly in noncommutative quantum mechanics
(NCQM), both in the Hilbert--Schmidt formulation
\cite{ScholtzChakraborty2013} and for generalized phase spaces with
coordinate and momentum noncommutativity \cite{LinHeng2022}; magnetic
spectral geometry for continuous magnetic operators has been developed in
a different setting \cite{de2020noncommutative}. We seek an exact
local-compactness criterion that separates this property from spectral
type and metric degeneration. The paper first establishes a structure
theorem for affine left--right Dirac data on a two-dimensional Moyal
plane, then derives a distinguished family from a fixed irreducible NCQM
sector and studies its affine \(U(1)_\star\) gauge coupling.

For a Moyal plane with deformation parameter \(\vartheta\ne0\), let
\(\nu=(\nu_1,\nu_2)^T\) and \(w=(w_1,w_2)^T\) be the left and right
affine Moyal coordinate symbols on \(T^*\mathbb R^2\), with
\[
\{\nu_1,\nu_2\}_{\mathrm{can}}=\vartheta,
\qquad
\{w_1,w_2\}_{\mathrm{can}}=-\vartheta,
\qquad
\{\nu_i,w_j\}_{\mathrm{can}}=0.
\]
For
\[
\eta=T_L\nu+T_Rw+c,
\qquad T_L,T_R\in M_2(\mathbb R),
\qquad c\in\mathbb R^2,
\]
we consider the two-Pauli Dirac operator
\[
D_\eta=\operatorname{Op}^W(\eta_1)\otimes\sigma_1
       +\operatorname{Op}^W(\eta_2)\otimes\sigma_2.
\]
The associated affine coordinate map is
\(\mathscr L_\eta:z\mapsto(\nu(z),\eta(z))\), with constant Jacobian
\(A_\eta:=d\mathscr L_\eta\).
Theorem~\ref{thm:affine-left-right} gives
\begin{equation}\label{eq:intro-affine-invariants}
\det A_\eta=\vartheta^2\det T_R,
\qquad
\Gamma=\vartheta(\det T_L-\det T_R),
\end{equation}
and proves the exact equivalence
\begin{equation}\label{eq:intro-local-compactness}
\boxed{\text{local compactness}\iff\det T_R\ne0.}
\end{equation}
In the locally compact regime, with \(\vartheta\) fixed, the right-block
determinant determines the leading localized heat coefficient, Dixmier-trace volume and oriented
Gaussian \(K\)-homology pairing. The scalar commutator \(\Gamma\)
determines the Landau or continuous spectral regime, while the
commutator seminorm depends on the full left block \(T_L\). Thus local
compactness, metric behaviour and spectral type are controlled by
different parts of the affine data.

We next specialize to a nondegenerate irreducible sector
\begin{equation}\label{eq:intro-sector}
\lambda_0=(\hbar_0,\vartheta_0,B_0),
\qquad
\hbar_0\vartheta_0B_0\ne0,
\qquad
\Delta:=\hbar_0-\vartheta_0B_0\ne0,
\end{equation}
of the seven-dimensional kinematical symmetry group
\(G_{\mathrm{NC}}\), a triple central extension of \(\mathbb R^4\)
\cite{chowdhuryali2013symmetry,chowdhuryali2014triply}. The labelled
central cocycles determine momentum and position planes in the quotient
by the centre, \(V=V_\Pi\oplus V_X\), and the momentum differential
calculus selects \(V_\Pi\). Restricting the inverse sector multiplier
form to its dual gives a bivector whose coefficient in the chosen basis is
\begin{equation}\label{eq:intro-theta-eff}
\vartheta_{\mathrm{eff}}
=\frac{\vartheta_0}{1-\vartheta_0B_0/\hbar_0}.
\end{equation}
The corresponding effective-coordinate plane commutes with the physical
position plane. A single metaplectic transport carries the
effective-coordinate pair and its commuting complement to the left--right
Moyal frame and simultaneously carries the normalized sector momenta to
the Dirac frame. The construction distinguishes
\[
(\hbar_0,\vartheta_0,B_0)\ \text{(sector)},
\qquad
(r,s)\ \text{(presentation)},
\qquad
\varrho\ \text{(ordering)}.
\]
The physical position plane also yields a locally compact Moyal triple,
with a generally different local volume; the main construction uses the
momentum-dual effective plane selected by the labelled sector data.

On this effective plane, affine gauge coupling acts as a shear of the
Dirac coordinates by the algebra coordinates: it changes only \(T_L\)
and leaves the sector-fixed right block \(T_R\) unchanged. The
two-parameter potential \(A_x=\alpha_xy\), \(A_y=\alpha_yx\)
illustrates this mechanism. Its curvature is
\[
F_{xy}=\alpha_y-\alpha_x
-\frac{e\vartheta_{\mathrm{eff}}}{\hbar_0}\alpha_x\alpha_y,
\]
and the coupled scalar commutator is
\begin{equation}\label{eq:intro-Gamma}
\Gamma(F_{xy})
=\frac12\left(1-\frac{eF_{xy}}{B_0}\right).
\end{equation}
Local compactness, the local spectral volume and the even
\(K\)-homology class therefore remain fixed throughout the affine gauge
family. Away from \(eF_{xy}=B_0\), the spectrum consists of infinitely
degenerate Landau levels. At magnetic cancellation the Dirac generators
commute but remain transverse to the algebra, and the spectrum becomes
\(\mathbb R\) with empty point spectrum.

Writing \(D_\infty\) for the coupled operator and
\(\pi^M_{2,\varrho}\) for the effective Moyal representation, the exact
localized heat trace is continuous through this transition and yields
local spectral dimension two and the measurable volume
\[
\operatorname{Tr}_\omega\!\left(
\pi^M_{2,\varrho}(a_\varrho)(1+D_\infty^2)^{-1}\right)
=2\left|\frac{\vartheta_0B_0}{\hbar_0}\right|
\tau_{\vartheta_{\mathrm{eff}}}(a_\varrho),
\]
where \(a_\varrho\in\mathcal A_{\vartheta_{\mathrm{eff}},\varrho}\)
and \(\tau_{\vartheta_{\mathrm{eff}}}\) is the normalized Moyal trace.
These conclusions extend to every real affine potential, including
arbitrary constant terms. In the two-parameter ansatz, the left block
becomes singular at \(eF_{xy}=-\Delta/\vartheta_0\), a locus disjoint
from spectral cancellation. Magnetic cancellation is also distinct from
the NCQM critical relation \(\vartheta_0B_0=\hbar_0\), or \(\Delta=0\),
which is excluded by sector nondegeneracy
\cite{NairPolychronakos2001,BellucciNersessianSochichiu2001}.

\paragraph{Relation to prior work.}
The standard Moyal Dirac operator belongs to the present two-generator
class, as recorded in Remark~\ref{rem:standard-moyal-affine-member}.
The harmonic Moyal triple of \cite{GayralWulkenhaar2013}, however, uses
a larger Clifford realization; in dimension two it acts on
\(L^2(\mathbb R^2)\otimes\mathbb C^4\) and lies outside the
two-Pauli-generator class considered here.
Wallet's two-spinor Landau--Moyal family \cite{wallet2012homothetic}
does belong to this class. At the purely right-regular LSZ point
\(\xi_{\mathrm W}=1\) \cite{LangmannSzaboZarembo2004}, the Lipschitz
seminorm vanishes while local compactness survives. At the purely
left-regular point \(\xi_{\mathrm W}=-1\), local compactness fails.
The latter contradicts the all-real \(\xi_{\mathrm W}\ne1\) extension
in \cite[Rem.~3.9]{wallet2012homothetic}, while Wallet's principal
result for \(\xi_{\mathrm W}>0\), \(\xi_{\mathrm W}\ne1\), is
unaffected. Remark~\ref{rem:wallet-affine-benchmark} gives the exact
left--right decomposition and both exceptional-point calculations.

The numerical combination
\(\vartheta_0/(1-\vartheta_0B_0/\hbar_0)\) has precedents in the
``exotic'' NCQM literature, where the corresponding ratio appears as
\(\theta/(1-B\theta)\)
\cite{DuvalHorvathy2000,DelducDuretGieresLefrancois2008}.
Here \(\vartheta_0\) already specifies the physical-position
commutator, whereas \(\vartheta_{\mathrm{eff}}\) governs the
complementary effective-coordinate plane. Algebraic decoupling into
two commuting planes is also standard
\cite{NairPolychronakos2001,BellucciNersessianSochichiu2001}.
Our construction derives the choice and normalization of the represented
momentum-dual plane from the labelled central cocycles and sector form,
and transports its algebraic frame jointly with the Dirac frame.

Localized Moyal heat asymptotics and spectral-action expansions were
developed by Gayral--Iochum \cite{GayralIochum2005}. Vassilevich's
localized noncommutative heat-kernel analysis \cite{Vassilevich2004}
and subsequent treatment of generalized Moyal-plane Laplacians with
both left and right regular multiplications \cite{Vassilevich2005}
provide further analytic antecedents. To our knowledge, the exact
criterion \(\det T_R\ne0\) for local compactness, together with the
simultaneous formulas for the localized heat coefficient, Dixmier volume,
spectral regime and oriented Gaussian pairing, has not previously been
isolated as a single left--right structure theorem. Its formulas remain
valid through the transition \(\Gamma=0\).

Section~\ref{sec:affine-left-right-general} proves the general affine
theorem and compares it with standard and Landau--Moyal examples.
Sections~\ref{sec:cstar-gnc}--\ref{sec:dirac-spectral-triple} construct
the fixed-sector algebra, differential calculus and sector spectral
triple. Section~\ref{sec:moyal-plane-realization} uses a different
represented algebra, generated by the effective coordinates, and
transports it jointly with the sector Dirac operator. It applies the
affine theorem to gauge coupling and treats the spectrum, heat trace,
volume, even \(K\)-homology class and cutoff removal.
Section~\ref{sec:conclusion} discusses further directions. The appendices
give the compression argument for the index pairing and the Weyl-phase
and cocycle normalizations used in the sector construction.

\section{Affine left--right Dirac geometry on the Moyal plane}
\label{sec:affine-left-right-general}
\label{subsec:affine-left-right-general}

We isolate a Moyal-geometric statement independent of the NCQM sector.  Fix
\(\vartheta\in\mathbb R\setminus\{0\}\) and \(\varrho\in\mathbb R\). We use the Fourier convention
\[
\widehat f(k):=\frac{1}{(2\pi)^2}
\int_{\mathbb R^2}e^{-ik\cdot u}f(u)\,du,
\qquad
f(u)=\int_{\mathbb R^2}e^{ik\cdot u}\widehat f(k)\,dk.
\]
For \(f,g\in\mathcal S(\mathbb R^2)\), the Weyl--Moyal product is
\[
(f\star_{\vartheta,1/2}g)(u)
=\int_{\mathbb R^2}\int_{\mathbb R^2}
\widehat f(k)\widehat g(\ell)
e^{\,i(k+\ell)\cdot u
-\frac{i\vartheta}{2}(k_1\ell_2-k_2\ell_1)}
\,dk\,d\ell.
\]
The Weyl involution is complex conjugation. This convention gives
\([x,y]_{\star_{\vartheta,1/2}}=i\vartheta\) for the coordinate multipliers.
Let \(\mathcal A_{\vartheta,\varrho}\) be the Schwartz Moyal algebra in the
\(\varrho\)-ordered realization, with ordering transport
\[
T_\varrho
:=\exp\!\left(i\left(\varrho-\frac12\right)
\vartheta\,\partial_x\partial_y\right).
\]
Thus
\[
\begin{aligned}
a_\varrho\star_{\vartheta,\varrho}b_\varrho
&=T_\varrho^{-1}\!\left((T_\varrho a_\varrho)
  \star_{\vartheta,1/2}(T_\varrho b_\varrho)\right),\\
a_\varrho^{*_{\varrho}}
&=T_\varrho^{-1}\!\left(\overline{T_\varrho a_\varrho}\right).
\end{aligned}
\]
The map \(T_\varrho:\mathcal A_{\vartheta,\varrho}
\to\mathcal A_{\vartheta,1/2}\) is a Fr\'echet-\(*\)-algebra isomorphism
and a unitary on \(L^2(\mathbb R^2)\).  For
\(a_\varrho\in\mathcal A_{\vartheta,\varrho}\), write
\[
a^{\mathrm W}:=T_\varrho a_\varrho,
\qquad
a_\varrho=T_\varrho^{-1}a^{\mathrm W}.
\]
In particular,
\[
T_\varrho L_{a_\varrho}^{(\varrho)}T_\varrho^{-1}
=L_{a^{\mathrm W}}^{(1/2)},
\qquad
\int_{\mathbb R^2}a^{\mathrm W}
=\int_{\mathbb R^2}a_\varrho.
\]
All left multipliers below are tensored with \(I_2\) when acting on
\(\mathcal H=L^2(\mathbb R^2)\otimes\mathbb C^2\); this tensor factor is
usually suppressed.

Let
\[
L=(L_x^{(\varrho)},L_y^{(\varrho)})^{\mathsf T},
\qquad
R=(R_x^{(\varrho)},R_y^{(\varrho)})^{\mathsf T}
\]
be the left and right coordinate multipliers.  Use the canonical Poisson
bracket on \(T^*\mathbb R^2\), with \(z=(x_1,x_2,\xi_1,\xi_2)\),
\[
\{f,g\}_{\mathrm{can}}
=\sum_{j=1}^2\left(
\partial_{x_j}f\,\partial_{\xi_j}g
-\partial_{\xi_j}f\,\partial_{x_j}g\right),
\qquad
J_2=\begin{pmatrix}0&1\\-1&0\end{pmatrix},
\qquad
J_4=\begin{pmatrix}0&I_2\\-I_2&0\end{pmatrix}.
\]
The Weyl symbols \(\nu=(\nu_1,\nu_2)^{\mathsf T}\) and
\(w=(w_1,w_2)^{\mathsf T}\) of \(L\) and \(R\) are
\begin{equation}\label{eq:general-left-right-symbols}
\begin{aligned}
\nu_1&=x_1-(1-\varrho)\vartheta\xi_2,
&\qquad w_1&=x_1+\varrho\vartheta\xi_2,\\
\nu_2&=x_2+\varrho\vartheta\xi_1,
&w_2&=x_2-(1-\varrho)\vartheta\xi_1.
\end{aligned}
\end{equation}
For affine-linear symbols, Weyl quantization gives exactly
\begin{equation}\label{eq:linear-symbol-commutator-exact}
[\operatorname{Op}^W(f),\operatorname{Op}^W(g)]
=i\operatorname{Op}^W(\{f,g\}_{\mathrm{can}})
\end{equation}
on Schwartz space.  Hence
\begin{equation}\label{eq:general-left-right-brackets}
\{\nu_1,\nu_2\}_{\mathrm{can}}=\vartheta,
\qquad
\{w_1,w_2\}_{\mathrm{can}}=-\vartheta,
\qquad
\{\nu_j,w_k\}_{\mathrm{can}}=0,
\end{equation}
equivalently \([L_x^{(\varrho)},L_y^{(\varrho)}]=i\vartheta I\),
\([R_x^{(\varrho)},R_y^{(\varrho)}]=-i\vartheta I\), and
\([L_\alpha^{(\varrho)},R_\beta^{(\varrho)}]=0\) for
\(\alpha,\beta\in\{x,y\}\).
The linear map \(z\mapsto(\nu,w)\) has determinant \(\vartheta^2\).

For \(T_L,T_R\in M_2(\mathbb R)\) and \(c\in\mathbb R^2\), put
\begin{equation}\label{eq:general-affine-eta}
\eta=T_L\nu+T_Rw+c.
\end{equation}
Here \(T_L\) and \(T_R\) are the coefficient matrices of the left and right
coordinate multipliers.  Define the affine coordinate map and its constant
Jacobian by
\[
\mathscr L_\eta:z\longmapsto(\nu(z),\eta(z)),
\qquad A_\eta:=d\mathscr L_\eta.
\]
Relative to \((\nu,w)\), the Jacobian has block matrix
\(\left(\begin{smallmatrix}I_2&0\\T_L&T_R\end{smallmatrix}\right)\).
Consequently \(\det A_\eta=\vartheta^2\det T_R\), so
\((\nu,\eta)\) are global affine coordinates precisely when \(T_R\) is
invertible.

On \(\mathcal S(\mathbb R^2)\otimes\mathbb C^2\), define
\begin{equation}\label{eq:general-affine-Dirac}
D_{T_L,T_R,c}
:=\operatorname{Op}^W(\eta_1)\otimes\sigma_1
 +\operatorname{Op}^W(\eta_2)\otimes\sigma_2,
\end{equation}
where the \(\sigma_j\) are the standard Pauli matrices. To verify essential self-adjointness, let \(N\) be the self-adjoint
closure of \(\left(
1-\partial_{x_1}^2-\partial_{x_2}^2+x_1^2+x_2^2
\right)\otimes I_2\). Then \(N\ge I\), and
Schwartz spinors form a core for \(N\). Both \(D=D_{T_L,T_R,c}\)
and \([D,N]\) are affine-linear combinations of the canonical
operators with constant matrix coefficients. The oscillator estimates
therefore give constants \(C_1,C_2>0\) such that
\[
\|D\psi\|\le C_1\|N\psi\|,
\qquad
\bigl|\langle D\psi,N\psi\rangle
-\langle N\psi,D\psi\rangle\bigr|
\le C_2\langle\psi,N\psi\rangle
\]
for every \(\psi\in\mathcal S(\mathbb R^2)\otimes\mathbb C^2\).
Nelson's commutator theorem
\cite[Theorem~X.37]{reed1975fourier}
therefore gives essential self-adjointness on Schwartz spinors. We use the
same notation for the self-adjoint closure.  The grading is
\(\chi=I\otimes\sigma_3\), and the chiral operator
\(D^+:\operatorname{Dom}D\cap\mathcal H^+\to\mathcal H^-\) is the
closure of the lower-left entry
\(\operatorname{Op}^W(\eta_1+i\eta_2)\).

The bracket relations give
\begin{equation}\label{eq:general-Gamma}
\Gamma:=\{\eta_1,\eta_2\}_{\mathrm{can}}
=\vartheta(\det T_L-\det T_R),
\end{equation}
and therefore
\begin{equation}\label{eq:general-Gamma-operator}
[\operatorname{Op}^W(\eta_1),\operatorname{Op}^W(\eta_2)]
=i\Gamma I.
\end{equation}
In particular,
\begin{equation}\label{eq:general-D-square}
D^2=H_\eta\otimes I_2-\Gamma I\otimes\sigma_3,
\qquad
H_\eta:=\operatorname{Op}^W(\eta_1)^2+
         \operatorname{Op}^W(\eta_2)^2,
\end{equation}
with the self-adjoint closures understood.  Moreover, with \(\partial_1:=\partial_x\) and \(\partial_2:=\partial_y\),
\[
[D,L_{a_\varrho}^{(\varrho)}]
=i\vartheta\sum_{j,k=1}^2(T_LJ_2)_{jk}
 L_{\partial_k a_\varrho}^{(\varrho)}\otimes\sigma_j
\]
is bounded.  Since left Schwartz multiplication preserves Schwartz spinors,
this identity also gives invariance of \(\operatorname{Dom}D\).

Let \(p_{1/2}(x,y)=2\exp[-(x^2+y^2)/|\vartheta|]\) be the centred
Gaussian minimal Weyl--Moyal projection, and put
\(p_\varrho=T_\varrho^{-1}p_{1/2}\).  Its integral is
\(2\pi|\vartheta|\), its class is the positive generator of \(K_0\), and
its left regular representation is a projection of infinite rank.

We write \(\mathcal Z_1=\mathcal M_{1,\infty}\) for the
Marcinkiewicz ideal, following Carey, Gayral, Rennie, and Sukochev (CGRS)
\cite[Sec.~2]{CareyGayralRennieSukochev2012}.
It consists of the compact operators \(K\) for which the Banach norm
\[
\|K\|_{1,\infty}
:=
\sup_{\substack{N\in\mathbb N\\N\ge1}}
\frac{1}{\log(1+N)}
\sum_{j=0}^{N-1}\mu_j(K).
\]
is finite, where \(\mu_j(K)\) are the singular values in
nonincreasing order. Let \(\mathcal Z_1^0\) denote the closure of the
trace-class operators in this norm.

Throughout this section, \(\operatorname{Tr}_\omega\) ranges over
Dixmier traces defined by dilation-invariant generalized limits;
measurability means independence of this choice.
All references to CGRS, including proposition numbering, refer
to the published Journal of Functional Analysis version \cite{CareyGayralRennieSukochev2012}.

\begin{theorem}[Affine left--right structure theorem for Moyal Dirac data]
\label{thm:affine-left-right}
Let \(D=D_{T_L,T_R,c}\) be the self-adjoint operator
\eqref{eq:general-affine-Dirac}, with representation
\(\pi_\varrho(a_\varrho)=L_{a_\varrho}^{(\varrho)}\otimes I_2\)
and grading \(\chi=I\otimes\sigma_3\).  Then:
\begin{enumerate}[label=(\roman*)]
\item The Jacobian \(A_\eta\) of the affine algebra--Dirac coordinate map \(\mathscr L_\eta\) satisfies
\begin{equation}\label{eq:general-det-transversality}
\boxed{\det A_\eta=\vartheta^2\det T_R,}
\end{equation}
and the scalar Dirac commutator is
\begin{equation}\label{eq:general-Gamma-det}
\boxed{\Gamma=\vartheta(\det T_L-\det T_R).}
\end{equation}

\item The represented affine Dirac data are locally compact if and only if
\begin{equation}\label{eq:general-local-compactness-iff}
\boxed{\det T_R\neq0.}
\end{equation}
In that case they form an even non-unital spectral triple, and
\[
L_{a_\varrho}^{(\varrho)}(1+D^2)^{-1/2}\in\mathcal K
\qquad(a_\varrho\in\mathcal A_{\vartheta,\varrho}).
\]
If \(\det T_R=0\), this operator is noncompact for every nonzero
\(a_\varrho\in\mathcal A_{\vartheta,\varrho}\).

\item If \(\det T_R\neq0\), then for every \(a_\varrho\) and \(t>0\),
\begin{equation}\label{eq:general-affine-heat}
\boxed{\operatorname{Tr}\!\left(
L_{a_\varrho}^{(\varrho)}e^{-tD^2}\right)
=\frac{1}{2\pi\vartheta^2|\det T_R|}
 \left(\int_{\mathbb R^2}a_\varrho\right)\Phi_\Gamma(t),}
\end{equation}
where
\begin{equation}\label{eq:general-affine-Phi}
\Phi_\Gamma(t)=
\begin{cases}
|\Gamma|\coth(|\Gamma|t),&\Gamma\neq0,\\[1mm]
t^{-1},&\Gamma=0.
\end{cases}
\end{equation}
The local spectral dimension is two, and, with the Dixmier-trace convention
specified above,
\begin{equation}\label{eq:general-affine-dixmier}
\boxed{\operatorname{Tr}_\omega\!\left(
L_{a_\varrho}^{(\varrho)}(1+D^2)^{-1}\right)
=\frac{1}{2\pi\vartheta^2|\det T_R|}
 \int_{\mathbb R^2}a_\varrho.}
\end{equation}
Equivalently, for
\(\tau_\vartheta(a_\varrho)=(2\pi|\vartheta|)^{-1}\int a_\varrho\),
the volume functional is
\begin{equation}\label{eq:general-affine-volume-canonical}
\frac{1}{|\vartheta|\,|\det T_R|}\,\tau_\vartheta(a_\varrho).
\end{equation}

\item If \(\det T_R\neq0\), the Gaussian generator has pairing
\begin{equation}\label{eq:general-affine-index}
\boxed{\left\langle[p_\varrho],[D_{T_L,T_R,c}]\right\rangle
=-\operatorname{sgn}(\vartheta)\operatorname{sgn}(\det T_R).}
\end{equation}
In particular, this orientation is independent of \(T_L\) and \(c\).

\item If \(\Gamma\neq0\), then
\begin{equation}\label{eq:general-affine-spectrum-Landau}
\operatorname{Spec}(D_{T_L,T_R,c})
=\{0\}\cup\{\,\pm\sqrt{2|\Gamma|n}:n\ge1\,\},
\end{equation}
and every eigenvalue has infinite multiplicity.  If \(\Gamma=0\) and
\(\det T_R\neq0\), then
\begin{equation}\label{eq:general-affine-spectrum-free}
\operatorname{Spec}(D_{T_L,T_R,c})=\mathbb R,
\qquad \sigma_{\mathrm p}(D_{T_L,T_R,c})=\varnothing.
\end{equation}
\end{enumerate}
\end{theorem}

The following lemma and analytic proposition supply the local estimates
for the theorem.  Its proof then establishes the compactness criterion,
the index pairing and the spectral alternatives.

For invertible \(A_\eta\), write
\(\Omega_\eta=A_\eta J_4A_\eta^{\mathsf T}\) for the Poisson matrix
in the \((\nu,\eta)\)-coordinates, and \(\#_{\Omega_\eta}\) for the
Weyl product in these coordinates.  We use the symbol classes
\[
|\partial_\nu^\alpha\partial_\eta^\beta q(\nu,\eta)|
\le C_{\alpha\beta}
\langle\nu\rangle^{m-|\alpha|}
\langle\eta\rangle^{n-|\beta|},
\qquad
q\in S^{m,n},
\qquad
S^{-\infty,n}:=\bigcap_{m\in\mathbb R}S^{m,n}.
\]
For a symbol expressed in these coordinates, \(\operatorname{Op}^W\)
means Weyl quantization after pullback by \(\mathscr L_\eta\).
The two orders in \(S^{m,n}\) are deliberately kept separate: the first
measures decay in the algebra variables \(\nu\), while the second measures
decay in the Dirac variables \(\eta\).  Left Moyal localization supplies
Schwartz decay in \(\nu\); the heat kernel and resolvent powers supply the
required decay in \(\eta\).  The transverse coordinate system therefore
turns local compactness and summability into a complementary-decay problem in
these two sets of variables.

\begin{lemma}[Complementary Weyl products]
\label{lem:general-complementary-weyl-trace}
Assume that \(A_\eta\) is invertible.  Let \(f(\nu,\eta)=f_0(\nu)\),
where \(f_0\in\mathcal S(\mathbb R^2)\), and let
\(g(\nu,\eta)=g_0(\eta)\) be matrix-valued, with
\[
|\partial^\beta g_0(\eta)|\le C_\beta\langle\eta\rangle^{n-|\beta|}.
\]
Then
\[
f\#_{\Omega_\eta}g,\ g\#_{\Omega_\eta}f\in S^{-\infty,n},
\qquad
f\#_{\Omega_\eta}g-g\#_{\Omega_\eta}f\in S^{-\infty,n-1}.
\]
If \(g_0\) is Schwartz, both products are Schwartz on \(\mathbb R^4\).
If \(n<-2\), their integrals satisfy, componentwise,
\begin{equation}\label{eq:general-complementary-integral}
\int_{\mathbb R^4}(f\#_{\Omega_\eta}g)(\nu,\eta)\,d\nu\,d\eta
=\left(\int_{\mathbb R^2}f_0\right)
 \left(\int_{\mathbb R^2}g_0\right).
\end{equation}
In particular, this identity holds for Schwartz \(g_0\); taking matrix
traces gives the corresponding scalar formula.
\end{lemma}

\begin{proof}
Let \(M_\eta\) be the \((\nu,\eta)\) block of \(\Omega_\eta\), and use
\(\widehat f_0(k)=(2\pi)^{-2}\int e^{-ik\cdot\nu}f_0(\nu)\,d\nu\).
The constant-coefficient Weyl product gives the exact formulas
\[
\begin{aligned}
(f\#_{\Omega_\eta}g)(\nu,\eta)
 &=\int_{\mathbb R^2}e^{ik\cdot\nu}\widehat f_0(k)
   g_0\!\left(\eta-\tfrac12M_\eta^{\mathsf T}k\right)\,dk,\\
(g\#_{\Omega_\eta}f)(\nu,\eta)
 &=\int_{\mathbb R^2}e^{ik\cdot\nu}\widehat f_0(k)
   g_0\!\left(\eta+\tfrac12M_\eta^{\mathsf T}k\right)\,dk.
\end{aligned}
\]
For every integer \(N\ge0\), integrate by parts using
\[
(1-\Delta_k)^N e^{ik\cdot\nu}
=\langle\nu\rangle^{2N}e^{ik\cdot\nu}.
\]
Derivatives in \(\nu\) contribute polynomial factors in \(k\),
while derivatives in \(\eta\) act on \(g_0\). Peetre's inequality
controls the shifted symbol weights by a power of \(\langle k\rangle\)
times the corresponding weight in \(\eta\); this power is absorbed by
the Schwartz decay of \(\widehat f_0\). Thus, for every pair of
multi-indices \(\alpha,\beta\),
\[
\begin{aligned}
&\bigl|\partial_\nu^\alpha\partial_\eta^\beta
(f\#_{\Omega_\eta}g)(\nu,\eta)\bigr|
+\bigl|\partial_\nu^\alpha\partial_\eta^\beta
(g\#_{\Omega_\eta}f)(\nu,\eta)\bigr|\\
&\hspace{2cm}\le
C_{N,\alpha,\beta}\langle\nu\rangle^{-2N}
\langle\eta\rangle^{n-|\beta|}.
\end{aligned}
\]
Since \(N\) is arbitrary, both products belong to \(S^{-\infty,n}\).
Writing \(h=\tfrac12M_\eta^Tk\), the identity
\[
g_0(\eta-h)-g_0(\eta+h)
=-\int_{-1}^{1}h\cdot\nabla g_0(\eta+t h)\,dt
\]
and the same estimates give order \(n-1\) in \(\eta\) for their
difference. If \(g_0\) is Schwartz, these estimates hold for every
\(n\), so both products are Schwartz on \(\mathbb R^4\).
Finally, when \(n<-2\), the preceding bounds give absolute
integrability in \((\nu,\eta)\). For each fixed \(\eta\), Fourier
inversion in \(\nu\) yields
\[
\int_{\mathbb R^2}(f\#_{\Omega_\eta}g)(\nu,\eta)\,d\nu
=(2\pi)^2\widehat f_0(0)g_0(\eta)
=\left(\int_{\mathbb R^2}f_0\right)g_0(\eta).
\]
Integrating in \(\eta\) proves the stated integral identity.
\end{proof}

\begin{proposition}[Localized heat trace, zeta residue, and volume]
\label{prop:general-affine-analytic-package}
Assume \(\det T_R\neq0\), and put \(D=D_{T_L,T_R,c}\) and
\(G=(1+D^2)^{-1}\).  For every
\(a_\varrho\in\mathcal A_{\vartheta,\varrho}\) and \(t>0\),
\(L_{a_\varrho}^{(\varrho)}e^{-tD^2}\) is trace class and the heat trace is
given by \eqref{eq:general-affine-heat}.  The localized zeta function is
holomorphic for \(\Re s>2\), extends meromorphically across \(s=2\) with at
most a simple pole, and satisfies
\begin{equation}\label{eq:general-package-zeta-residue}
\operatorname*{Res}_{s=2}
\operatorname{Tr}\!\left(L_{a_\varrho}^{(\varrho)}(1+D^2)^{-s/2}\right)
=\frac{1}{\pi\vartheta^2|\det T_R|}\int_{\mathbb R^2}a_\varrho.
\end{equation}
Finally, \(L_{a_\varrho}^{(\varrho)}G\in\mathcal Z_1\) is measurable and
its Dixmier trace is given by \eqref{eq:general-affine-dixmier}.
\end{proposition}

\begin{proof}
Choose \(m\in\mathbb R^2\) with \(\vartheta T_RJ_2m=c\), and let
\(U_m=\exp(i\,m\cdot R)\).  The right Weyl group commutes with the
represented left algebra and translates \(T_LL+T_RR+c\) by
\(-\vartheta T_RJ_2m\).  Hence
\begin{equation}\label{eq:general-c-removal}
U_mD_{T_L,T_R,c}U_m^*=D_{T_L,T_R,0},
\qquad [U_m,L_{a_\varrho}^{(\varrho)}]=0.
\end{equation}
We may therefore set \(c=0\).

The localization symbol is \(a^{\mathrm W}(\nu)\), with
\(a^{\mathrm W}=T_\varrho a_\varrho\).  If \(\Gamma\neq0\), the
normalized pair
\[
\frac{\operatorname{Op}^W(\eta_1)}{\sqrt{|\Gamma|}},
\qquad
\frac{\operatorname{sgn}(\Gamma)\operatorname{Op}^W(\eta_2)}
     {\sqrt{|\Gamma|}}
\]
has commutator \(iI\).  The Mehler formula and
\eqref{eq:general-D-square} give the matrix-valued heat symbol
\[
k_t(\eta)=
\frac{e^{t\Gamma\sigma_3}}{\cosh(|\Gamma|t)}
\exp\!\left[-\frac{\tanh(|\Gamma|t)}{|\Gamma|}|\eta|^2\right].
\]
If \(\Gamma=0\), the independent linear observables have commuting Weyl
groups and hence strongly commute; in this case
\(k_t(\eta)=e^{-t|\eta|^2}I_2\).  In both cases \(k_t\) is Schwartz for
\(t>0\), and
\begin{equation}\label{eq:general-package-gaussian-integral}
\int_{\mathbb R^2}\operatorname{tr}_{\mathbb C^2}k_t(\eta)\,d\eta
=2\pi\Phi_\Gamma(t).
\end{equation}
Lemma~\ref{lem:general-complementary-weyl-trace} makes the localized heat
symbol Schwartz, so its Weyl operator is trace class.  The Weyl trace
formula \cite[Thm.~4.4.21]{nicola2010global}, the Jacobian
\(|\det A_\eta|^{-1}=(\vartheta^2|\det T_R|)^{-1}\), and
\(\int a^{\mathrm W}=\int a_\varrho\) now prove
\eqref{eq:general-affine-heat}.  The two branches agree as \(\Gamma\to0\).

For the resolvent powers, let \(\Gamma_{\mathrm E}\) denote Euler's gamma
function.  For \(\Re \gamma>0\), the heat representation gives
\[
G^\gamma=\operatorname{Op}^W(g_\gamma(\eta)),
\qquad
g_\gamma(\eta)=\frac{1}{\Gamma_{\mathrm E}(\gamma)}
\int_0^\infty t^{\gamma-1}e^{-t}k_t(\eta)\,dt.
\]
To estimate these symbols, write
\[
a_t:=
\begin{cases}
\tanh(|\Gamma|t)/|\Gamma|,&\Gamma\ne0,\\
t,&\Gamma=0.
\end{cases}
\]
For fixed \(\Gamma\), one has \(a_t\asymp t\) on \(0<t\le1\),
and the matrix prefactor of \(k_t\) is uniformly bounded. Differentiating
the Gaussian therefore gives, for every multi-index \(\beta\),
\[
\bigl|\partial_\eta^\beta k_t(\eta)\bigr|
\le C_\beta t^{|\beta|/2}
e^{-b_\Gamma t|\eta|^2},
\qquad 0<t\le1,
\]
with \(b_\Gamma>0\). Consequently, for \(\operatorname{Re}\gamma>0\),
\[
\begin{aligned}
\left|
\partial_\eta^\beta
\int_0^1t^{\gamma-1}e^{-t}k_t(\eta)\,dt
\right|
&\le C_\beta
\int_0^1t^{\operatorname{Re}\gamma-1+|\beta|/2}
e^{-b_\Gamma t|\eta|^2}\,dt\\
&\le C_{\gamma,\beta}
\langle\eta\rangle^{-2\operatorname{Re}\gamma-|\beta|}.
\end{aligned}
\]
On \([1,\infty)\), the factor \(e^{-t}\) and the Gaussian estimates
give a Schwartz remainder. Including the factor
\(\Gamma_{\mathrm E}(\gamma)^{-1}\), we obtain
\[
\bigl|\partial_\eta^\beta g_\gamma(\eta)\bigr|
\le C_{\gamma,\beta}
\langle\eta\rangle^{-2\operatorname{Re}\gamma-|\beta|},
\]
locally uniformly in \(\gamma\) on \(\operatorname{Re}\gamma>0\).
Differentiation with respect to \(\gamma\) introduces powers of
\(\log t\), which are absorbed by the positive margin
\(\operatorname{Re}\gamma-a\) on compact subsets of
\(\operatorname{Re}\gamma>a>0\). Thus the family is holomorphic
in every fixed symbol class of \(\eta\)-order \(-2a\) on this half-plane.

For \(\Re s>2\), Lemma~\ref{lem:general-complementary-weyl-trace} gives
\[
L_{a_\varrho}^{(\varrho)}G^{s/2}
\in\operatorname{Op}^W(S^{-\infty,-\Re s}).
\]
The symbol and its derivatives are integrable, so the derivative-\(L^1\)
Weyl criterion \cite[Thm.~4.4.21]{nicola2010global} gives trace-class
membership and trace-norm holomorphy. Applying the integral identity of
Lemma~\ref{lem:general-complementary-weyl-trace} to \(g_{s/2}\)
and using its heat representation gives the Mellin formula,
valid for \(\Re s>2\):
\begin{equation}\label{eq:mellin-formula}
  \operatorname{Tr}\!\left(L_{a_\varrho}^{(\varrho)}G^{s/2}\right)
=\frac{C_{a_\varrho}}{\Gamma_{\mathrm E}(s/2)}
 \int_0^\infty t^{s/2-1}e^{-t}\Phi_\Gamma(t)\,dt,
\qquad
C_{a_\varrho}:=
\frac{\int_{\mathbb R^2}a_\varrho}{2\pi\vartheta^2|\det T_R|}.  
\end{equation}
The interchange is absolutely convergent in this half-plane.  Since
\(e^{-t}\Phi_\Gamma(t)=t^{-1}-1+O(t)\) as \(t\downarrow0\),
subtracting \(t^{-1}\) on \((0,1)\) gives at most a simple pole at
\(s=2\), with residue \(2C_{a_\varrho}\).  This proves
\eqref{eq:general-package-zeta-residue} for arbitrary \(a_\varrho\).

It remains to prove the Dixmier-trace statement.
We first place represented Schwartz elements in the
zeta-integrability algebra \(B_\zeta(G)\).
By \cite[Proposition~3.8]{CareyGayralRennieSukochev2012},
the sandwiched operators \(L_bGL_d\) then belong to
\(\mathcal Z_1\).
Moyal factorization, together with the trace-class commutator
estimate below, gives \(L_aG\in\mathcal Z_1\) for every
algebra element.
We then show that the resulting Dixmier functional is a
positive continuous trace, whose value on one Gaussian minimal
projection determines it on the whole rapidly decreasing
matrix algebra. The preceding symbol estimate gives \(G=\operatorname{Op}^W(g_1(\eta))\), with no
\(\nu\)-dependence.  The commutator assertion of
Lemma~\ref{lem:general-complementary-weyl-trace} therefore gives
\[
[L_{b_\varrho}^{(\varrho)},G]
\in\operatorname{Op}^W(S^{-\infty,-3})
\subset\mathcal L^1\subset\mathcal Z_1^0.
\]
For clarity, suppress the ordering superscript on represented multipliers
in the rest of this proof.  If
\(q_\varrho=b_\varrho\star_{\vartheta,\varrho}b_\varrho^{*_{\varrho}}\),
then \(L_{q_\varrho}=L_{b_\varrho}L_{b_\varrho}^*\ge0\).
Ordinary trace cyclicity, the Mellin formula~\eqref{eq:mellin-formula}, and
\(\Phi_\Gamma(t)\le t^{-1}+|\Gamma|\) imply
\[
\operatorname{Tr}(L_{b_\varrho}^*G^uL_{b_\varrho})
\le C_{q_\varrho}\left(\frac{1}{u-1}+|\Gamma|\right),
\qquad 1<u\le2.
\]
The same bound applies to \(b_\varrho^{*_{\varrho}}\), proving
\(L_{b_\varrho}\in B_\zeta(G)\).

The Moyal factorization property
\(\mathcal A_{\vartheta,\varrho}
=\mathcal A_{\vartheta,\varrho}\star_{\vartheta,\varrho}
 \mathcal A_{\vartheta,\varrho}\)
follows from \cite[Prop.~2.7]{gayral2004moyal} by ordering transport.
Proposition~3.8 of \cite{CareyGayralRennieSukochev2012} gives \(L_{b_\varrho}G L_{d_\varrho}\in\mathcal Z_1\).
Together with
\[
L_{b_\varrho}L_{d_\varrho}G
=L_{b_\varrho}G L_{d_\varrho}
 +L_{b_\varrho}[L_{d_\varrho},G],
\]
this proves \(L_{a_\varrho}G\in\mathcal Z_1\) for every
\(a_\varrho\).  Cyclicity of the Dixmier trace and the trace-class
commutators show that
\(F_\omega(a_\varrho):=\operatorname{Tr}_\omega(L_{a_\varrho}G)\)
is tracial.  It is positive on squares, since
\[
F_\omega(d_\varrho^{*_{\varrho}}\star_{\vartheta,\varrho}d_\varrho)
=\operatorname{Tr}_\omega(L_{d_\varrho}^*G L_{d_\varrho})\ge0.
\]
The continuity needed here follows from the closed graph theorem.  Indeed,
suppose \(a_{\varrho,n}\to a_\varrho\) in the Schwartz Fr\'echet
topology and \(L_{a_{\varrho,n}}G\to K\) in \(\mathcal Z_1\).
Continuity of the left regular representation gives
\(L_{a_{\varrho,n}}G\to L_{a_\varrho}G\) in operator norm.  On the other
hand, the continuous embedding \(\mathcal Z_1\hookrightarrow\mathcal K\)
implies \(L_{a_{\varrho,n}}G\to K\) in operator norm as well.  Hence
\(K=L_{a_\varrho}G\), so the graph of
\(a_\varrho\mapsto L_{a_\varrho}G\) is closed.  Since the domain is
Fr\'echet and \((\mathcal Z_1,\|\cdot\|_{1,\infty})\) is Banach, the closed graph theorem makes this map continuous; consequently \(F_\omega\) is continuous.

Set \(P=L_{p_\varrho}\). The hypotheses of
\cite[Proposition~4.12]{CareyGayralRennieSukochev2012}
hold for \(a=P\) and \(\delta=1/2\), since
\[
P^{1-\delta}=P\in B_\zeta(G),
\qquad
[G,P^{1-\delta}]=[G,P]
\in\mathcal L^1\subset\mathcal Z_1^0.
\]
The zeta residue gives
\[
\lim_{u\downarrow1}(u-1)\operatorname{Tr}(PG^uP)
=C_{p_\varrho}
=\frac{1}{|\vartheta|\,|\det T_R|}.
\]
That proposition makes \(PGP\) measurable with this value for every
Dixmier trace.  Since \(PG-PGP=P[G,1-P]\) is trace class, the same
holds for \(PG\).
Finally, the Schwartz Moyal algebra is Fr\'echet-isomorphic to the rapidly
decreasing matrix algebra \cite[Prop.~2.5]{gayral2004moyal}.  Its matrix
units have \(p_\varrho\) as a diagonal representative.  Traciality
annihilates the off-diagonal units and gives the same value on every
diagonal unit; continuity extends this identity to the full algebra.
Thus
\[
F_\omega(a_\varrho)
=F_\omega(p_\varrho)\,
  \frac{1}{2\pi|\vartheta|}\int_{\mathbb R^2}a_\varrho,
\]
which proves \eqref{eq:general-affine-dixmier} and measurability.
\end{proof}

\begin{proof}[Proof of Theorem~\ref{thm:affine-left-right}]
Part (i) is the Jacobian and bracket calculation preceding the theorem.
For the sufficiency in (ii), Proposition~\ref{prop:general-affine-analytic-package}
shows that, for \(\varepsilon>0\),
\[
K_\varepsilon:=L_{a_\varrho}^{(\varrho)}e^{-\varepsilon D^2}
                         (1+D^2)^{-1/2}
\]
is compact.  Functional calculus gives
\[
\begin{aligned}
\|L_{a_\varrho}^{(\varrho)}(1+D^2)^{-1/2}-K_\varepsilon\|
&\le\|L_{a_\varrho}^{(\varrho)}\|
 \sup_{\lambda\ge0}
 \frac{1-e^{-\varepsilon\lambda}}{\sqrt{1+\lambda}}\\
&\le\|L_{a_\varrho}^{(\varrho)}\|\sqrt{\varepsilon}
\longrightarrow0.
\end{aligned}
\]
Local compactness follows.  Self-adjointness, grading, domain invariance,
and bounded commutators were established above, and the left regular
representation is faithful.

Conversely, if \(\det T_R=0\), choose \(k\neq0\) with \(T_RJ_2k=0\).
The group \(U_t=\exp(it\,k\cdot R)\) commutes with the left algebra
and with both Dirac generators, hence with
\[
K_{a_\varrho}:=L_{a_\varrho}^{(\varrho)}(1+D^2)^{-1/2}.
\]
The generator \(k\cdot R\) is a nonzero linear canonical observable
and has purely absolutely continuous spectrum, so \(U_t\) has no
nonzero finite-dimensional invariant subspace.  If \(K_{a_\varrho}\)
were nonzero and compact, a positive eigenvalue of
\(K_{a_\varrho}^*K_{a_\varrho}\) would give such a subspace.
For \(a_\varrho\neq0\), the operator \(K_{a_\varrho}\) is nonzero
because \((1+D^2)^{-1/2}\) has dense range and the representation is
faithful.  This proves necessity, including the asserted failure for every
nonzero algebra element.

Part (iii) follows from Proposition~\ref{prop:general-affine-analytic-package}.
For a nonzero positive square the \(t^{-1}\) heat coefficient is strictly
positive, so the critical local summability exponent is exactly two.

For (iv), ordering covariance reduces to Weyl ordering.  Under the
Wigner identification
\(L^2(\mathbb R^2)\simeq\mathcal H_{\mathrm L}\otimes\mathcal H_{\mathrm R}\),
the Gaussian projection is \(P=|h_0\rangle\langle h_0|\otimes I_{\mathcal H_{\mathrm R}}\),
with the spinor factor suppressed.  It preserves \(\operatorname{Dom}D\)
and has bounded commutator with \(D\); hence the compression bridge,
Lemma~\ref{lem:compression-bridge}, identifies the pairing with the index
of the compressed chiral operator.  Since \(PL_x^{(1/2)}P=PL_y^{(1/2)}P=0\) and
\eqref{eq:general-c-removal} removes \(c\) while commuting with \(P\),
we may unitarily identify each of \(P\mathcal H^+\) and
\(P\mathcal H^-\) with \(\mathcal H_{\mathrm R}\). Under these
identifications, set
\[
R_1:=R_x^{(1/2)},\qquad R_2:=R_y^{(1/2)},\qquad
(T_RR)_j:=\sum_{k=1}^2 (T_R)_{jk}R_k,\quad j=1,2.
\]
Thus \((T_RR)_j\) is the \(j\)-th component of the operator-valued
column \(T_RR\).  Then
\begin{equation}\label{eq:general-compressed-chiral-TR}
PD^+P\simeq C:=(T_RR)_1+i(T_RR)_2.
\end{equation}
Writing \(\kappa=-\vartheta\det T_R\), one has
\[
[(T_RR)_1,(T_RR)_2]=i\kappa I_{\mathcal H_{\mathrm R}},
\qquad [C,C^*]=2\kappa I_{\mathcal H_{\mathrm R}}.
\]
On this single canonical factor, a metaplectic transformation identifies
\(C\), up to a positive scalar, with an annihilation operator if
\(\kappa>0\), and with a creation operator if \(\kappa<0\).
Their indices are \(+1\) and \(-1\), respectively, proving
\eqref{eq:general-affine-index}.

For (v), if \(\Gamma\neq0\), a Weyl displacement removes the constants
and a metaplectic transformation reduces the two Dirac generators to a
canonical pair \((Q_{\mathrm c},P_{\mathrm c})\),
\([Q_{\mathrm c},P_{\mathrm c}]=iI\), with
\begin{equation}\label{eq:general-landau-normal-form}
D\simeq\sqrt{|\Gamma|}\,
\bigl(Q_{\mathrm c}\otimes\sigma_1
+\operatorname{sgn}(\Gamma)P_{\mathrm c}\otimes\sigma_2\bigr)
\otimes I_{\mathcal H_{\mathrm{mult}}}.
\end{equation}
Here \(\mathcal H_{\mathrm{mult}}\) is the Hilbert space of the unused
canonical pair and gives the infinite-dimensional multiplicity factor.
The square on the active factor is
\[
|\Gamma|(Q_{\mathrm c}^2+P_{\mathrm c}^2)\otimes I_2
-\Gamma I\otimes\sigma_3,
\]
with eigenvalues \(2|\Gamma|n\), \(n\ge0\).  The nonzero levels have
the two chiral partners, giving both signs of
\(\sqrt{2|\Gamma|n}\); the multiplicity factor makes every eigenvalue
infinitely degenerate.

If \(\Gamma=0\) and \(\det T_R\neq0\), first remove \(c\) by
\eqref{eq:general-c-removal}.  The two remaining linear forms are
independent and Poisson commute.  Completing them to a symplectic basis
and applying the corresponding metaplectic unitary gives
\begin{equation}\label{eq:general-free-normal-form}
UDU^{-1}=M_{x_1}\otimes\sigma_1+M_{x_2}\otimes\sigma_2
\quad\text{on }L^2(\mathbb R^2,\mathbb C^2).
\end{equation}
The fibre eigenvalues are \(\pm|x|\), whose essential ranges together
give \(\mathbb R\).  An eigenvector of eigenvalue \(\lambda\) would be
supported on \(\{|x|=|\lambda|\}\), a set of Lebesgue measure zero.
Thus the point spectrum is empty.
\end{proof}

\begin{remark}[The standard Moyal Dirac operator]
\label{rem:standard-moyal-affine-member}
For \(D=-i\sigma^\mu\partial_\mu\), the identities
\[
\xi_1=\frac{\nu_2-w_2}{\vartheta},
\qquad \xi_2=\frac{w_1-\nu_1}{\vartheta}
\]
give \(T_L^{\mathrm{std}}:=\vartheta^{-1}J_2\) and \(T_R^{\mathrm{std}}:=-\vartheta^{-1}J_2\).
Hence \(\Gamma=0\), \(\det A_\eta=1\), and
\[
\operatorname{Tr}_\omega\!\left(
L_{a_{1/2}}^{(1/2)}(1+D^2)^{-1}\right)
=\frac{1}{2\pi}\int_{\mathbb R^2}a_{1/2}.
\]
Thus the standard Moyal spectral triple is a balanced member
\(\det T_L=\det T_R\) of the affine class.
\end{remark}

\begin{remark}[Wallet's Landau--Moyal family and the two one-sided points]
\label{rem:wallet-affine-benchmark}
At Weyl ordering, Wallet's convention is \(\Theta=\vartheta J_2\) and
\(\widetilde x=2\Theta^{-1}x\).  Conjugating his Clifford matrices by
\(\sigma_1\) converts his second Pauli matrix to ours.  This preserves
spectra and commutator norms, but reverses the grading; the comparisons
below concern compactness and metric seminorms.
For his operator \cite[Sec.~3]{wallet2012homothetic}, expressed in our
Pauli convention,
\[
\mathcal D_{\xi_{\mathrm W}}
=-i\sigma^\mu\partial_\mu
 +\xi_{\mathrm W}\sigma^\mu m(\widetilde x_\mu),
\]
the identities, which follow from \cite[Eq.~(2.4)]{wallet2012homothetic},
\[
(-i\partial_\mu+\widetilde x_\mu)\psi=\psi\star\widetilde x_\mu,
\qquad
(-i\partial_\mu-\widetilde x_\mu)\psi=-\widetilde x_\mu\star\psi
\]
give, with \(L_\mu=L(\widetilde x_\mu)\) and
\(R_\mu=R(\widetilde x_\mu)\),
\begin{equation}\label{eq:wallet-exact-left-right}
\mathcal D_{\xi_{\mathrm W}}
=\frac{1+\xi_{\mathrm W}}2\,\sigma^\mu R_\mu
 -\frac{1-\xi_{\mathrm W}}2\,\sigma^\mu L_\mu.
\end{equation}
Consequently its coefficient matrices are
\begin{equation}\label{eq:wallet-TLTR}
T_L^{\mathrm W}=\frac{1-\xi_{\mathrm W}}{\vartheta}J_2,
\qquad
T_R^{\mathrm W}=-\frac{1+\xi_{\mathrm W}}{\vartheta}J_2.
\end{equation}
Here the superscript \(\mathrm W\) labels Wallet's family.  In particular,
\[
\det T_L^{\mathrm W}=\frac{(1-\xi_{\mathrm W})^2}{\vartheta^2},
\qquad
\det T_R^{\mathrm W}=\frac{(1+\xi_{\mathrm W})^2}{\vartheta^2},
\qquad
|\Gamma_{\mathrm W}|=\frac{4|\xi_{\mathrm W}|}{|\vartheta|}.
\]
The commutator identity
\[
[\mathcal D_{\xi_{\mathrm W}},L(a_{1/2})]
=(1-\xi_{\mathrm W})L(-i\partial_\mu a_{1/2})\otimes\sigma^\mu
\]
gives, for \(\xi_{\mathrm W}\neq1\), the distance homothety factor
\[
\left|\frac{\det T_L^{\mathrm{std}}}{\det T_L^{\mathrm W}}\right|^{1/2}
=\frac{1}{|1-\xi_{\mathrm W}|}.
\]
The two one-sided limits can be checked directly on the regular
representation, without a heat-kernel calculation.

\smallskip
\noindent\textit{The right-regular point \(\xi_{\mathrm W}=1\).}
Here \(T_L^{\mathrm W}=0\), \(T_R^{\mathrm W}=-2\vartheta^{-1}J_2\),
and \(\mathcal D_1=\sigma^\mu R_\mu\).  At this LSZ point
\cite{LangmannSzaboZarembo2004}, the Lipschitz seminorm vanishes, so the
spectral distance between any two distinct states is infinite, consistently
with Wallet's metric discussion \cite[Rems.~3.2,~3.5]{wallet2012homothetic}.
Under the Wigner identification
\[
L^2(\mathbb R^2)\simeq\mathcal H_{\mathrm L}\otimes\mathcal H_{\mathrm R},
\qquad
L(a_{1/2})\simeq A_{a_{1/2}}\otimes I_{\mathcal H_{\mathrm R}},
\]
Wallet's scalar Landau Hamiltonian
\(H_L=\sum_\mu R(\widetilde x_\mu)^2\) is represented as
\(I_{\mathcal H_{\mathrm L}}\otimes h_{\mathrm R}\), where \(h_{\mathrm R}\) is a one-mode harmonic-oscillator Hamiltonian.  Hence
\[
L(a_{1/2})(H_L+1)^{-1}
\simeq A_{a_{1/2}}\otimes(h_{\mathrm R}+1)^{-1}.
\]
Both factors are Hilbert--Schmidt: \(A_{a_{1/2}}\) has a Schwartz Weyl
symbol \cite[Thm.~1.30]{folland1989harmonic}, and \(h_{\mathrm R}\) has
linearly growing eigenvalues.
  The localized Landau resolvent is therefore
Hilbert--Schmidt, and Theorem~\ref{thm:affine-left-right} gives local
compactness of the Dirac data.  This corrects the non-Hilbert--Schmidt
conclusion in the proof of Wallet's Proposition~3.3.  His
Proposition~3.3(ii), Corollary~3.4 and Proposition~3.6 exclude this point;
they do not themselves assert failure of local compactness there.

\smallskip
\noindent\textit{The left-regular point \(\xi_{\mathrm W}=-1\).}
Here \(T_L^{\mathrm W}=2\vartheta^{-1}J_2\), \(T_R^{\mathrm W}=0\),
and \(\mathcal D_{-1}=-\sigma^\mu L_\mu\).  
Under the same Wigner identification, Wallet's scalar Landau Hamiltonian
\(H_L=\sum_\mu L(\widetilde x_\mu)^2\) is represented as
\(h_{\mathrm L}\otimes I_{\mathcal H_{\mathrm R}}\), where
\(h_{\mathrm L}\) is a one-mode harmonic-oscillator Hamiltonian.  Hence
\[
L(a_{1/2})(H_L+1)^{-1}
\simeq A_{a_{1/2}}(h_{\mathrm L}+1)^{-1}\otimes I_{\mathcal H_{\mathrm R}}.
\]
For nonzero \(a_{1/2}\),
\(A_{a_{1/2}}(h_{\mathrm L}+1)^{-1}\neq0\), since \(A_{a_{1/2}}\neq0\)
and \((h_{\mathrm L}+1)^{-1}\) has dense range.  Its tensor product with the
identity on the infinite-dimensional right factor therefore cannot be compact.

Theorem~\ref{thm:affine-left-right} likewise gives failure of Dirac local
compactness.  This contradicts the extension in Wallet's Remark~3.9 to
all real \(\xi_{\mathrm W}\neq1\), which includes \(-1\).
These exceptional-point corrections leave his main Proposition~3.6 for
\(\xi_{\mathrm W}>0\), \(\xi_{\mathrm W}\neq1\), unaffected.
\end{remark}

\begin{remark}[Non-affine left perturbations]
\label{rem:nonaffine-vdd}
The triangular map \((\nu,\eta)\mapsto(\nu,\eta+A(\nu))\) does not
by itself extend the affine theorem to an unbounded nonlinear Moyal
potential: self-adjointness and global symbol estimates require additional
control.  A sufficient framework is van den Dungen's locally bounded
perturbation theorem \cite[Theorem~4.2]{vanDenDungen2018}.
Starting from a transverse affine triple, let \(M\) be a densely defined,
closed, symmetric odd perturbation satisfying
\[
\pi_\varrho(a_\varrho)\operatorname{Dom}M\subset\operatorname{Dom}M,
\qquad
M\pi_\varrho(a_\varrho)\ \text{extends to a bounded operator}
\]
for every \(a_\varrho\in\mathcal A_{\vartheta,\varrho}\).
Assume an adequate approximate identity \((e_n)\) for \(D\) in this
algebra with
\(\sup_n\|[M,\pi_\varrho(e_n)]\|<\infty\).
Since the left regular representation is nondegenerate, that theorem
applies: the closure of \(D+M\) defines an even locally compact spectral
triple with the same \(KK\)-class.  It supplies no exact heat-trace or
spectral formula for the nonlinear perturbation.
\end{remark}

\section{\texorpdfstring{$C^*$}{C*}-algebras associated with the kinematical symmetry group \texorpdfstring{$G_{\mathrm{NC}}$}{GNC}}
\label{sec:cstar-gnc}

Fix the NCQM sector
\[
\lambda_0=(\hbar_0,\vartheta_0,B_0),\qquad
\hbar_0\vartheta_0B_0\ne0,\qquad
\Delta:=\hbar_0-\vartheta_0B_0\ne0.
\]
We identify its twisted group \(C^*\)-algebra, its faithful sector
representations, and its place among the central-character fibres of
\(C^*(G_{\mathrm{NC}})\).  The presentation parameters
\(r\ne\hbar_0/(\vartheta_0B_0)\) and \(s\in\mathbb R\) change the
concrete realization of the selected fibre.

\subsection{The Weyl system and twisted convolution}
\label{subsec:construction-cstar}

The four strongly continuous one-parameter unitary groups of the sector
representation \cite[Theorem~V.1]{chowdhury2017representations} act by
\begin{equation}\label{eq:unitary_actions}
\begin{aligned}
(U_1(q_1)\phi)(x,y)
&=e^{-\frac{iB_0(1-r)}{r\vartheta_0B_0-\hbar_0}q_1y}
\phi\!\left(x+
\frac{\vartheta_0B_0(r+s-rs)-\hbar_0}{r\vartheta_0B_0-\hbar_0}q_1,y\right),\\
(U_2(q_2)\phi)(x,y)
&=e^{-\frac{irB_0}{\hbar_0}q_2x}
\phi\!\left(x,y-
\frac{r\vartheta_0B_0(1-s)-\hbar_0}{\hbar_0}q_2\right),\\
(U_3(q_3)\phi)(x,y)
&=e^{\frac{i}{\hbar_0}q_3x}
\phi\!\left(x,y-s\frac{\vartheta_0}{\hbar_0}q_3\right),\\
(U_4(q_4)\phi)(x,y)
&=e^{\frac{i}{\hbar_0}q_4y}
\phi\!\left(x+(1-s)\frac{\vartheta_0}{\hbar_0}q_4,y\right).
\end{aligned}
\end{equation}
Their generators are \((\Pi_x,\Pi_y,X,Y)/\hbar_0\), in the notation
of \eqref{eq:noncentral-generators}.  The shift--multiplication formulas
give
\begin{equation}\label{eq:all_weyl_commutation_relation}
\begin{aligned}
U_1(q_1)U_3(q_3)&=e^{iq_1q_3/\hbar_0}U_3(q_3)U_1(q_1),\\
U_2(q_2)U_4(q_4)&=e^{iq_2q_4/\hbar_0}U_4(q_4)U_2(q_2),\\
U_1(q_1)U_2(q_2)&=e^{-iB_0q_1q_2/\hbar_0}U_2(q_2)U_1(q_1),\\
U_3(q_3)U_4(q_4)&=e^{-i\vartheta_0q_3q_4/\hbar_0^2}U_4(q_4)U_3(q_3),\\
[U_1(q_1),U_4(q_4)]&=[U_2(q_2),U_3(q_3)]=0.
\end{aligned}
\end{equation}
Equivalently,
\begin{equation}\label{eq:weyl_commutation_relation}
U_k(q_k)U_j(q_j)
=e^{2\pi i\tau_{jk}q_jq_k}U_j(q_j)U_k(q_k),
\end{equation}
where
\begin{equation}\label{eq:skewsymmatrix}
\tau=\tau_{\hbar_0,\vartheta_0,B_0}
=\frac1{2\pi\hbar_0}
\begin{pmatrix}
0&B_0&-1&0\\
-B_0&0&0&-1\\
1&0&0&\vartheta_0/\hbar_0\\
0&1&-\vartheta_0/\hbar_0&0
\end{pmatrix}.
\end{equation}
In particular, the commutation phases are independent of \((r,s)\).
Appendix~\ref{App:weyl} records a representative calculation.

For \(q\in\mathbb R^4\), set
\(U(q)=U_1(q_1)U_2(q_2)U_3(q_3)U_4(q_4)\).  Reordering the factors
gives
\begin{equation}\label{eq:U-product}
U(q)U(p)
=\exp\!\left(2\pi i\sum_{j<k}p_j\tau_{jk}q_k\right)U(q+p).
\end{equation}
The symmetrized system is
\begin{equation}\label{eq:projective_representation}
\begin{aligned}
W_{\lambda_0}^{r,s}(q)
&:=e^{\pi i\sum_{j<k}q_j\tau_{jk}q_k}U(q)\\
&=\exp\!\left[\frac{i}{\hbar_0}
(q_1\Pi_x+q_2\Pi_y+q_3X+q_4Y)\right].
\end{aligned}
\end{equation}
Here and below the exponential uses the self-adjoint closure of the real
linear canonical observable.  Suppressing \((r,s)\), one has
\begin{equation}\label{eq:sector-weyl-multiplier}
W_{\lambda_0}(q)W_{\lambda_0}(p)
=\omega_{\lambda_0}(q,p)W_{\lambda_0}(q+p),\qquad
W_{\lambda_0}(q)^*=W_{\lambda_0}(-q),
\end{equation}
with the explicitly normalized cocycle
\begin{equation}\label{eq:sector-cocycle}
\omega_{\lambda_0}(q,p)
:=\omega_{\hbar_0,\vartheta_0,B_0}(q,p)
=e^{-\pi i q^{\mathsf T}\tau p}
=e^{-\frac i2q^{\mathsf T}\Omega_{\lambda_0}p},
\qquad \Omega_{\lambda_0}=2\pi\tau.
\end{equation}
Bilinearity gives the cocycle identity, while skew-symmetry gives
\(\omega_{\lambda_0}(q,-q)=1\).

Define twisted convolution directly, for \(f,g\in L^1(\mathbb R^4)\), by
\begin{equation}\label{eq:star_product}
\begin{aligned}
(f\star_{\lambda_0}g)(p)
&:=\int f(q)g(p-q)\omega_{\lambda_0}(q,p-q)\,dq\\
&=\int f(q)g(p-q)e^{\pi i p^{\mathsf T}\tau q}\,dq,
\end{aligned}
\end{equation}
where \(\star_{\lambda_0}=\star_{\hbar_0,\vartheta_0,B_0}\), and set
\begin{equation}\label{eq:involution}
f^*(q)=\overline{f(-q)}.
\end{equation}

\begin{remark}[The twisted Banach algebra]
\label{rem:twisted-banach-star}
The cocycle identity and Fubini's theorem give associativity and
\begin{equation}\label{eq:star-involution}
(f\star_{\lambda_0}g)^*=g^*\star_{\lambda_0}f^*.
\end{equation}
Since \(|\omega_{\lambda_0}|=1\), the usual integral norm
\begin{equation}\label{eq:L1_norm}
\|f\|_1=\int_{\mathbb R^4}|f(q)|\,dq
\end{equation}
satisfies
\begin{equation}\label{eq:L1_submultiplicative}
\|f\star_{\lambda_0}g\|_1\le\|f\|_1\|g\|_1,
\qquad \|f^*\|_1=\|f\|_1.
\end{equation}
Thus these operations make \(L^1(\mathbb R^4)\) a Banach
\(*\)-algebra, denoted by \(L^1(\mathbb R^4,\omega_{\lambda_0})\).
\end{remark}

\begin{proposition}[Integrated sector representation]
\label{prop:twisted_banach_algebra_representation}
For each admissible \((r,s)\), define
\begin{equation}\label{eq:twisted_banach_algerba_representation}
\rho^{r,s}_{\hbar_0,\vartheta_0,B_0}(f)
:=\int_{\mathbb R^4} f(-q)W_{\lambda_0}^{r,s}(q)\,dq .
\end{equation}
Then \(\rho^{r,s}_{\hbar_0,\vartheta_0,B_0}\) is a contractive
nondegenerate \(*\)-representation of
\(L^1(\mathbb R^4,\omega_{\lambda_0})\).
\end{proposition}

\begin{proof}
The integrated map \(f\mapsto\int f(q)W_{\lambda_0}^{r,s}(q)\,dq\)
is a contractive \(*\)-representation by the multiplier law and Fubini.
The reflection \(f(q)\mapsto f(-q)\) is an isometric
\(*\)-automorphism because
\(\omega_{\lambda_0}(-q,-p)=\omega_{\lambda_0}(q,p)\), which proves the
same assertions for \(\rho^{r,s}_{\hbar_0,\vartheta_0,B_0}\).
A normalized approximate identity with shrinking support converges strongly
to the identity under the integrated representation, proving nondegeneracy.
\end{proof}

\subsection{Completion and the smooth represented fibre}

\begin{proposition}[Universal completion]
\label{prop:universal-cstar}
The formula
\begin{equation}\label{eq:star_norm}
\|f\|_*:=\sup_\pi\|\pi(f)\|,
\end{equation}
where \(\pi\) ranges over nondegenerate \(*\)-representations of
\(L^1(\mathbb R^4,\omega_{\lambda_0})\), defines a faithful
\(C^*\)-norm, with \(\|f\|_*\le\|f\|_1\).  Its completion is
\begin{equation}\label{eq:sector-cstar-compacts}
C^*(\mathbb R^4,\omega_{\lambda_0})
\cong\mathcal K(L^2(\mathbb R^2)).
\end{equation}
\end{proposition}

\begin{proof}
Contractivity of \(*\)-representations gives the upper bound, and the
\(C^*\)-identity holds under the supremum.  To show that it is a norm,
use the twisted left regular representation on \(L^2(\mathbb R^4)\),
\(\lambda_\omega(f)\xi=f\star_{\lambda_0}\xi\).  Young's inequality
gives boundedness, and this representation is nondegenerate.  If
\(\lambda_\omega(f)=0\), then
\(f\star_{\lambda_0}\eta_\epsilon=0\) for every normalized
\(C_c^\infty\) approximate identity.  Twisted translation continuity
gives \(f\star_{\lambda_0}\eta_\epsilon\to f\) in \(L^1\), hence
\(f=0\).  Thus \(\|\cdot\|_*\) is faithful.

The form is nondegenerate, since
\begin{equation}\label{eq:sector-tau-determinant}
\det\tau=\frac{\Delta^2}{16\pi^4\hbar_0^6}\ne0.
\end{equation}
A Darboux change of variables reduces its Weyl relations to the canonical
ones.  The integrated-form correspondence and the finite-dimensional
Stone--von Neumann theorem identify every nondegenerate representation
with a multiple of the Schr\"odinger representation
\cite{folland1989harmonic}.  In that representation,
integration of Schwartz functions gives all Schwartz-kernel operators:
this follows from a partial Fourier transform and an invertible linear
change of variables in the kernel formula.  Their norm closure is the
compact operators, since they are compact and include the rank-one
operators with Schwartz vectors.  Schwartz density in \(L^1\) now
proves \eqref{eq:sector-cstar-compacts}.
\end{proof}

\begin{corollary}[Faithful sector realizations]
\label{cor:rho-faithful-extension}
For each admissible \((r,s)\), the integrated sector representation
extends faithfully and nondegenerately to
\(C^*(\mathbb R^4,\omega_{\lambda_0})\), with image exactly
\(\mathcal K(L^2(\mathbb R^2))\).  In particular,
\[
\|\rho^{r,s}_{\hbar_0,\vartheta_0,B_0}(f)\|=\|f\|_*.
\]
The realizations for different \((r,s)\) are unitarily equivalent.
\end{corollary}

\begin{proof}
The extensions exist by universality and are nonzero.  Simplicity of the
compact operators makes them faithful.  The sector Weyl systems are
irreducible, so Stone--von Neumann gives the asserted image and unitary
equivalence.
\end{proof}

\begin{remark}[The Schwartz algebra and its image]
\label{rem:schwartz-dense}\label{rem:image-cstar}
\label{rem:image-schwartz-dense}
Set \(\mathcal S_{\hbar_0,\vartheta_0,B_0}=\mathcal S(\mathbb R^4)\)
with the twisted operations above.  It is a Fr\'echet
\(*\)-algebra: differentiation of the phase in
\eqref{eq:star_product} introduces only polynomial factors, and the
identity \(p=q+(p-q)\) distributes polynomial weights between the two
Schwartz factors.  These estimates also prove joint continuity of the
product.  Density in \(L^1\), together with
\(\|\cdot\|_*\le\|\cdot\|_1\), gives density in the completion.
Consequently, writing \(\rho=\rho^{r,s}_{\hbar_0,\vartheta_0,B_0}\),
\[
\mathcal C^*_{\mathcal B}
:=\rho\bigl(C^*(\mathbb R^4,\omega_{\lambda_0})\bigr)
=\mathcal K(L^2(\mathbb R^2)),\qquad
\mathscr S_{\mathcal B}:=\rho(\mathcal S_{\hbar_0,\vartheta_0,B_0})
\]
is a dense represented Schwartz \(*\)-subalgebra of \(\mathcal C^*_{\mathcal B}\). The \(C^*\)-image is already closed, so no additional operator-norm closure is required.
\end{remark}

\subsection{Central-character fibres of the group algebra}
\label{subsec:central-character-fibres}

Let \(N=Z(G_{\mathrm{NC}})\simeq\mathbb R^3\), with labelled central
basis \((Z_\hbar,Z_\vartheta,Z_B)\).  In the translation-parameter
ordering \((\Pi_x,\Pi_y,X,Y)\), the extension cocycles are
\[
\begin{aligned}
\xi(q,p)&=\tfrac12(q_1p_3+q_2p_4-q_3p_1-q_4p_2),\\
\xi'(q,p)&=\tfrac12(q_3p_4-q_4p_3),\qquad
\xi''(q,p)=\tfrac12(q_1p_2-q_2p_1),
\end{aligned}
\]
with the normalization of Appendix~\ref{App:matrix}; see also
\cite[Theorem~IV.1]{chowdhuryali2013symmetry}.

\begin{proposition}[Central-character fibres]
\label{prop:cstar-bundle}
For \(\kappa=(\kappa_\hbar,\kappa_\vartheta,\kappa_B)
\in\widehat N\simeq\mathbb R^3\), put
\[
\beta_\kappa=\kappa_\hbar\xi+
\kappa_\vartheta\xi'+\kappa_B\xi'',\qquad
\omega_\kappa(q,p)=e^{i\beta_\kappa(q,p)},\qquad
R_\kappa=\operatorname{rad}\beta_\kappa.
\]
Then \(C^*(G_{\mathrm{NC}})\) is canonically a
\(C_0(\widehat N)\)-algebra whose fibre at \(\kappa\) is
\begin{equation}\label{eq:cstar-fibre-central-character}
C^*(G_{\mathrm{NC}})_\kappa
\cong C^*(\mathbb R^4,\omega_\kappa).
\end{equation}
Choosing a complement to \(R_\kappa\) gives the noncanonical
identification
\begin{equation}\label{eq:degenerate-fibre-structure}
C^*(\mathbb R^4,\omega_\kappa)
\cong C_0(R_\kappa^*)\otimes
\mathcal K\!\left(L^2\!\left(
\mathbb R^{(4-\dim R_\kappa)/2}\right)\right).
\end{equation}
In particular, the form is nondegenerate exactly when
\begin{equation}\label{eq:kappa-nondegeneracy}
\kappa_\hbar^2-\kappa_\vartheta\kappa_B\ne0,
\end{equation}
and then the fibre is \(\mathcal K(L^2(\mathbb R^2))\).
\end{proposition}

\begin{proof}
The central inclusion induces a nondegenerate homomorphism
\[
C^*(N)=C_0(\widehat N)\longrightarrow ZM(C^*(G_{\mathrm{NC}})),
\]
where \(ZM\) denotes the centre of the multiplier algebra.
In the fibre at \(\kappa\), the central subgroup acts by its
character \(\gamma_\kappa(z)=e^{i\kappa\cdot z}\).  The canonical
section of the quotient by \(N\) then gives precisely the projective
law with multiplier \(\omega_\kappa\).  Conversely, every strongly
continuous \(\omega_\kappa\)-projective representation lifts to a
representation of \(G_{\mathrm{NC}}\) with this central character.
The universal properties therefore give
\eqref{eq:cstar-fibre-central-character}; this is the central-extension
construction of \cite{packer1992structure}.

Choose \(\mathbb R^4=R_\kappa\oplus W_\kappa\).  The cocycle is
trivial on the radical and has no mixed terms, while its restriction to
\(W_\kappa\) is nondegenerate.  Fourier transformation on
\(R_\kappa\) and Stone--von Neumann on \(W_\kappa\) give
\eqref{eq:degenerate-fibre-structure}.  The Pfaffian calculation in
Appendix~\ref{App:matrix} gives \eqref{eq:kappa-nondegeneracy}.
\end{proof}

On the open dense chart \(\kappa_\hbar\ne0\), write
\begin{equation}\label{eq:central-character-physical-chart}
\kappa_\hbar=\frac1\hbar,\qquad
\kappa_\vartheta=-\frac{\vartheta}{\hbar^2},\qquad
\kappa_B=-\frac B\hbar.
\end{equation}
The cocycle formulas give
\(\omega_\kappa=\omega_{\hbar,\vartheta,B}\).  Thus the fixed
physical sector corresponds to
\[
\kappa_0=\left(\frac1{\hbar_0},
-\frac{\vartheta_0}{\hbar_0^2},-\frac{B_0}{\hbar_0}\right),\qquad
\kappa_{0,\hbar}^2-\kappa_{0,\vartheta}\kappa_{0,B}
=\frac{\Delta}{\hbar_0^3}\ne0.
\]
It therefore selects exactly the compact-operator fibre represented
above.  Section~\ref{sec:dirac-spectral-triple} equips this fibre with
its differential and Dirac data.

\section{Dirac operator and the sector spectral triple}
\label{sec:dirac-spectral-triple}

Fix the sector \(\lambda_0=(\hbar_0,\vartheta_0,B_0)\), with
\(\hbar_0\vartheta_0B_0\ne0\) and
\(\Delta=\hbar_0-\vartheta_0B_0\ne0\), and an admissible pair
\(r\in\mathbb R\setminus\{\hbar_0/(\vartheta_0B_0)\}\),
\(s\in\mathbb R\).  We use the faithful sector representation
\(\rho=\rho^{r,s}_{\hbar_0,\vartheta_0,B_0}\) on
\(\mathcal H_0=L^2(\mathbb R^2)\), and write
\[
\begin{gathered}
\mathcal A=\mathcal S_{\hbar_0,\vartheta_0,B_0},\qquad
\mathcal C=C^*(\mathbb R^4,\omega_{\hbar_0,\vartheta_0,B_0}),\\
\mathcal C^*_{\mathcal B}=\rho(\mathcal C)=\mathcal K(\mathcal H_0),
\qquad \mathscr S_{\mathcal B}=\rho(\mathcal A).
\end{gathered}
\]
The compact-operator identification follows from the irreducible sector realization in
Section~\ref{sec:cstar-gnc}.

Define the physical canonical momenta by
\(p_x=-i\hbar_0\partial_x\), \(p_y=-i\hbar_0\partial_y\).
Suppressing the fixed sector and presentation labels, the noncentral
generators are
\begin{equation}\label{eq:noncentral-generators}
\begin{aligned}
X&=x-s\frac{\vartheta_0}{\hbar_0}p_y,
&Y&=y+(1-s)\frac{\vartheta_0}{\hbar_0}p_x,\\
\Pi_x&=\frac{(1-r)\hbar_0B_0}{\hbar_0-r\vartheta_0B_0}\,y
+\frac{(r+s-rs)\vartheta_0B_0-\hbar_0}
 {r\vartheta_0B_0-\hbar_0}\,p_x,\\
\Pi_y&=-rB_0x+
\left(1+r(s-1)\frac{\vartheta_0B_0}{\hbar_0}\right)p_y.
\end{aligned}
\end{equation}
Their quotients by \(\hbar_0\) generate the one-parameter groups in
\eqref{eq:unitary_actions}.  We use the self-adjoint closures, with
Schwartz space as a core.  On this core,
\begin{equation}\label{eq:sector-operator-CCR}
\begin{gathered}
[X,\Pi_x]=[Y,\Pi_y]=i\hbar_0I,\qquad
[X,Y]=i\vartheta_0I,\qquad
[\Pi_x,\Pi_y]=i\hbar_0B_0I,\\
[X,\Pi_y]=[Y,\Pi_x]=0.
\end{gathered}
\end{equation}

\subsection{The action and its differential calculus}
\label{subsec:differential_structure}

For \(q\in\mathbb R^4\), put
\(q^\sharp=(B_0q_2-q_3,-B_0q_1-q_4)\).  Initially on the twisted
\(L^1\)-algebra, define
\begin{equation}\label{eq:automorphism}
(\alpha_k f)(q)=e^{-ik\cdot q^\sharp}f(q),
\qquad k\in\mathbb R^2.
\end{equation}
Linearity of \(q\mapsto q^\sharp\) makes \(\alpha_k\) an
isometric \(*\)-automorphism of the twisted \(L^1\)-algebra.  It
preserves the universal norm and therefore extends to \(\mathcal C\).
Dominated convergence in \(L^1\), followed by density and
\(\|\cdot\|_*\le\|\cdot\|_1\), proves pointwise norm continuity
of this \(\mathbb R^2\)-action on \(\mathcal C\).

Let \(\partial_j\) be the closed infinitesimal generator of
\(t\mapsto\alpha_{te_j}\), with its norm-differentiability domain.
For \(a\in\mathcal A\),
\begin{equation}\label{eq:derivatives}
(\partial_j a)(q)=-iq_j^\sharp a(q),\qquad j=1,2,
\end{equation}
so \(\mathcal A\) is invariant under every iterated derivative and
\begin{equation}\label{eq:containment}
\mathcal A\subset
\bigcap_{n,m\ge0}\operatorname{Dom}(\partial_1^n\partial_2^m).
\end{equation}

The implementing unitaries are
\(U_k=e^{-i(k_1\Pi_x+k_2\Pi_y)}\), where
\(k_1\Pi_x+k_2\Pi_y\) denotes the self-adjoint closure of its
Schwartz-space restriction.  They obey the projective law
\begin{equation}\label{eq:projective-momentum-implementers}
U_kU_\ell
=e^{-\frac i2\hbar_0B_0(k_1\ell_2-k_2\ell_1)}U_{k+\ell}.
\end{equation}
The scalar multiplier in this projective law cancels under conjugation.  Hence
\begin{equation}\label{eq:star_automorphism}
\beta_k(A):=U_kAU_k^*,\qquad A\in\mathcal C^*_{\mathcal B},
\end{equation}
is an ordinary \(\mathbb R^2\)-action.  We use \(\beta\) for the
represented action to distinguish it from \(\alpha\).

\begin{proposition}[Covariance and derivations]
\label{prop:sector-covariance-derivations}
The action \(\beta\) is pointwise norm continuous, and
\begin{equation}\label{eq:sector-action-covariance}
\rho(\alpha_k f)=\beta_k(\rho(f)),\qquad f\in\mathcal C.
\end{equation}
If \(\delta_j\) is its \(j\)-th infinitesimal generator, then
\[
\operatorname{Dom}\delta_j=\rho(\operatorname{Dom}\partial_j),
\qquad \delta_j\rho(f)=\rho(\partial_j f).
\]
The generator domains are \(*\)-subalgebras, and for
\(f,g\in\operatorname{Dom}\partial_j\),
\[
\partial_j(f\star g)=(\partial_jf)\star g+f\star(\partial_jg),
\qquad \partial_j(f^*)=(\partial_jf)^*.
\]
The two derivations commute on the smooth vectors of the action,
including \(\mathcal A\); the analogous statements hold for
\(\delta_1,\delta_2\).
\end{proposition}

\begin{proof}
For the symmetrized sector Weyl system \(W(q)=W_{\lambda_0}^{r,s}(q)\), the Weyl relations give
\[
U_kW(q)U_k^*=e^{ik\cdot q^\sharp}W(q).
\]
Integration against \(f(-q)\) proves
\eqref{eq:sector-action-covariance} on \(L^1\), and continuity extends
it to \(\mathcal C\).  Since \(\rho\) is an isometric isomorphism
onto its image, norm continuity and the generator-domain identities
follow.  Differentiating the product and involution in norm proves the
derivation identities.  Mixed derivatives commute on smooth orbit maps.
On \(\mathcal A\), this also follows directly from multiplication
by the two linear functions in \eqref{eq:derivatives}.
\end{proof}

\begin{proposition}[Commutator domains]
\label{prop:sector-commutator-domains}
Write \(\Pi_1=\Pi_x\), \(\Pi_2=\Pi_y\).  Then
\begin{equation}\label{eq:sector-commutator-domain}
\begin{aligned}
\operatorname{Dom}\delta_j=\{A\in\mathcal C^*_{\mathcal B}:{}&
A\operatorname{Dom}\Pi_j\subset\operatorname{Dom}\Pi_j,\\
&[\Pi_j,A]\text{ extends to an element of }
\mathcal C^*_{\mathcal B}\},
\end{aligned}
\end{equation}
and
\begin{equation}\label{eq:star_derivations_of_A}
\delta_1(A)=-i[\Pi_x,A],\qquad
\delta_2(A)=-i[\Pi_y,A]
\end{equation}
on their respective domains.  Moreover
\(\mathscr S_{\mathcal B}\subset
\operatorname{Dom}\delta_1\cap\operatorname{Dom}\delta_2\).
\end{proposition}

\begin{proof}
Fix \(j\in\{1,2\}\). If \(A\in\operatorname{Dom}\delta_j\) and
\(\psi\in\operatorname{Dom}\Pi_j\), then
\[
e^{-it\Pi_j}A\psi
=\beta_{te_j}(A)e^{-it\Pi_j}\psi.
\]
The right-hand side is differentiable at \(t=0\). Stone's theorem
therefore implies \(A\psi\in\operatorname{Dom}\Pi_j\), and
differentiation gives
\[
-i\Pi_jA\psi=\delta_j(A)\psi-iA\Pi_j\psi.
\]
Thus \(A\) preserves \(\operatorname{Dom}\Pi_j\), and
\([\Pi_j,A]\) extends to
\(i\delta_j(A)\in\mathcal C^*_{\mathcal B}\).

Conversely, suppose \(A\in\mathcal C^*_{\mathcal B}\) preserves
\(\operatorname{Dom}\Pi_j\), and let
\(B\in\mathcal C^*_{\mathcal B}\) be the bounded extension of
\([\Pi_j,A]\). Differentiating on \(\operatorname{Dom}\Pi_j\),
integrating and extending by density gives
\[
\beta_{te_j}(A)-A
=-i\int_0^t\beta_{se_j}(B)\,ds.
\]
The integral converges in operator norm by norm continuity of
\(s\mapsto\beta_{se_j}(B)\). Dividing by \(t\) and letting
\(t\to0\) proves \(A\in\operatorname{Dom}\delta_j\) and
\(\delta_j(A)=-iB\).

Finally, the covariance formula of
Proposition~\ref{prop:sector-covariance-derivations}
shows that every element of \(\mathscr S_{\mathcal B}\)
has a smooth orbit map under \(\beta\), proving the
asserted inclusion.
\end{proof}

\subsection{The normalized Dirac operator and its spectrum}
\label{subsec:dirac_operator}

Set
\[
c_0:=\sqrt{|\hbar_0B_0|},\qquad
\varepsilon_0:=\operatorname{sgn}(\hbar_0B_0).
\]
The normalized momenta
\begin{equation}\label{eq:normalized-noncentral-generators}
\widetilde\Pi_x^{\,r,s}=\Pi_x^{r,s}/c_0,
\qquad
\widetilde\Pi_y^{\,r,s}=\varepsilon_0\Pi_y^{r,s}/c_0
\end{equation}
satisfy
\begin{equation}\label{eq:commutation}
[\widetilde\Pi_x,\widetilde\Pi_y]=iI.
\end{equation}
On \(\mathcal H=\mathcal H_0\otimes\mathbb C^2\), use the Pauli
matrices and grading \(\chi=I\otimes\sigma_3\) of
Section~\ref{sec:affine-left-right-general}.

\begin{definition}[Sector Dirac operator]
\label{def:3.1}
On Schwartz spinors define
\begin{equation}\label{eq:sector-dirac-data}
D=D^{r,s}_{\hbar_0,\vartheta_0,B_0}
:=\frac1{\sqrt2}
(\widetilde\Pi_x\otimes\sigma_1+
\widetilde\Pi_y\otimes\sigma_2).
\end{equation}
It is symmetric on this domain, since its coefficients and Pauli matrices
are self-adjoint and act on different tensor factors.
\end{definition}

\begin{proposition}[Self-adjointness and oscillator normal form]
\label{prop:3.4}\label{lem:spectral-multiplicity}
The sector Dirac operator is essentially self-adjoint on Schwartz
spinors.  Denoting its closure by the same symbol, one has a metaplectic
normal form, with tensor factors reordered,
\begin{equation}\label{eq:sector-oscillator-normal-form}
D\simeq D_{\mathrm{osc}}\otimes I_{L^2(\mathbb R)},\qquad
D_{\mathrm{osc}}=\frac1{\sqrt2}
(Q\otimes\sigma_1+P\otimes\sigma_2),\qquad [Q,P]=iI.
\end{equation}
Its spectrum is
\begin{equation}\label{eq:sector-dirac-spectrum}
\operatorname{Spec}(D)=\{0\}\cup
\{\pm\sqrt n:n=1,2,\ldots\},
\end{equation}
with every spectral value an eigenvalue of infinite multiplicity.
In particular, \(D\) has no compact resolvent.
\end{proposition}

\begin{proof}
The generators \eqref{eq:noncentral-generators} are real linear canonical
observables.  The self-adjointness argument of
Section~\ref{sec:affine-left-right-general} therefore applies, and
\eqref{eq:commutation} gives the Dirac commutator \(\Gamma=1/2\).
The metaplectic normal-form argument used in the proof of
Theorem~\ref{thm:affine-left-right}(v), applied to the canonical pair
\eqref{eq:commutation}, gives \eqref{eq:sector-oscillator-normal-form} and the
stated spectrum and multiplicities.
\end{proof}

\subsection{The sector spectral triple}
\label{subsec:spectral_triple}

Represent the Schwartz algebra on spinors by
\begin{equation}\label{eq:sector-spectral-data}
\pi^{r,s}(a):=\rho^{r,s}_{\hbar_0,\vartheta_0,B_0}(a)\otimes I_2,
\qquad a\in\mathcal A,\qquad
\mathcal H=L^2(\mathbb R^2)\otimes\mathbb C^2.
\end{equation}

\begin{lemma}[Represented Schwartz operators]
\label{lem:rho-schwartz-compact}
For \(a\in\mathcal A\), the operator \(\rho(a)\) has a Schwartz
kernel and maps \(\mathcal H_0\) continuously into
\(\mathcal S(\mathbb R^2)\).  In particular,
\(\pi^{r,s}(a)\) is trace class and maps \(\mathcal H\) into Schwartz
spinors.
\end{lemma}

\begin{proof}
The sector generators in \eqref{eq:noncentral-generators} are real
linear combinations of the standard canonical operators, and their
commutation matrix is nondegenerate because \(\Delta\ne0\).
After a linear Darboux normalization, they therefore form a canonical
frame related to the standard one by a linear symplectic transformation.
The corresponding intertwiner \(U\) can be chosen metaplectic.

As in the proof of Proposition~\ref{prop:universal-cstar}, integration
of the normalized Weyl system gives
\[
U\rho(a)U^{-1}=\operatorname{Op}^W(b),
\qquad b\in\mathcal S(\mathbb R^4).
\]
Metaplectic operators and their inverses preserve Schwartz space.
Moreover, metaplectic covariance of Weyl quantization shows that
\(\rho(a)\) itself has a Schwartz Weyl symbol, hence a Schwartz kernel
in the original coordinates.

Cauchy--Schwarz applied to each differentiated, polynomially weighted
kernel gives the continuous map
\(\mathcal H_0\to\mathcal S(\mathbb R^2)\).
The Schwartz Weyl symbol also implies trace-class membership
\cite[Thm.~4.4.21]{nicola2010global}.
Tensoring with \(I_2\) preserves these properties.
\end{proof}

\begin{theorem}[The sector spectral triple]
\label{thm:3.3.1}
For every admissible \((r,s)\), the data
\[
\left(\mathcal A,\mathcal H,\pi^{r,s},
D^{r,s}_{\hbar_0,\vartheta_0,B_0},\chi\right)
\]
form an even locally compact non-unital spectral triple.  The full
represented triples are unitarily equivalent as \((r,s)\) varies.
\end{theorem}

\begin{proof}
Self-adjointness and the Schwartz core follow from
Proposition~\ref{prop:3.4}.  The preceding lemma gives
\(\pi^{r,s}(a)\operatorname{Dom}D\subset\operatorname{Dom}D\).
Covariance and \eqref{eq:star_derivations_of_A} give
\begin{equation}\label{eq:bounded_commutator_spectral_triple}
[D,\pi^{r,s}(a)]
=\frac{i}{\sqrt2c_0}\rho(\partial_1a)\otimes\sigma_1
+\frac{i\varepsilon_0}{\sqrt2c_0}
\rho(\partial_2a)\otimes\sigma_2,
\end{equation}
initially on Schwartz spinors and then on \(\operatorname{Dom}D\).
It is bounded because \(\partial_ja\in\mathcal A\).
Moreover \(\pi^{r,s}(a)\) is already compact, so
\(\pi^{r,s}(a)(1+D^2)^{-1/2}\) is compact.  Faithfulness follows
from Corollary~\ref{cor:rho-faithful-extension}, and the Pauli relations
give \(D\chi=-\chi D\) and \([\pi^{r,s}(a),\chi]=0\).

An intertwiner between two irreducible sector Weyl systems intertwines
their integrated representations and, by Stone's theorem, their
momentum generators and domains.  Tensoring it with \(I_2\) therefore
intertwines the Dirac operators and gradings as well.
\end{proof}

\begin{remark}[The two represented geometries]
\label{rem:sector-vs-effective-triples}
In this sector representation, every Schwartz algebra element is already
trace class.  Thus \(\pi^{r,s}(a)(1+D^2)^{-t/2}\) is trace class for every \(t\ge0\); the local summability threshold is zero.  The even
\(KK\)-class is also zero.  Indeed, choose a rank-one projection with
Schwartz range in the represented algebra.  Its spinor amplification
has rank one on each chiral subspace, so the compressed chiral operator
is a map \(\mathbb C\to\mathbb C\) and has index zero.
Lemma~\ref{lem:compression-bridge} identifies this index with the
pairing, and Morita equivalence
\(KK^0(\mathcal K,\mathbb C)\cong\mathbb Z\) shows that the class vanishes.

Section~\ref{sec:moyal-plane-realization} uses a different represented
algebra: the effective-coordinate algebra acts by left Moyal
multiplication, and its minimal projections have infinite rank.  That
triple has local spectral dimension two and a nonzero generator class.
The common sector and Dirac frame do not identify these two represented
geometries; Proposition~\ref{prop:sector-left-regular} supplies the joint
transport for the effective plane.
\end{remark}

\section{\texorpdfstring{Moyal geometry, affine gauge coupling, noncommutative volume, and \(K\)-homology}{Moyal geometry, affine gauge coupling, noncommutative volume, and K-homology}}
\label{sec:moyal-plane-realization}

We construct the effective Moyal representation and transport the sector
Dirac operator with the same unitary.  The affine theorem of
Section~\ref{sec:affine-left-right-general} then gives the spectral and
geometric consequences of the external gauge coupling.
Throughout,
\[
\hbar_0\vartheta_0B_0\ne0,\qquad
\Delta:=\hbar_0-\vartheta_0B_0\ne0,\qquad
\vartheta_{\mathrm{eff}}:=\frac{\hbar_0\vartheta_0}{\Delta},\qquad
c_0:=\sqrt{|\hbar_0B_0|},\quad
\varepsilon_0:=\operatorname{sgn}(\hbar_0B_0).
\]
We suppress the sector and presentation superscripts on
\(Z=(X,Y,\Pi_x,\Pi_y)^{\mathsf T}\) when no ambiguity arises.  Its
commutation matrix is
\begin{equation}\label{eq:sigma-matrix}
[Z_i,Z_j]=i\Sigma_{ij}I,
\qquad
\Sigma=\begin{pmatrix}
0&\vartheta_0&\hbar_0&0\\
-\vartheta_0&0&0&\hbar_0\\
-\hbar_0&0&0&\hbar_0B_0\\
0&-\hbar_0&-\hbar_0B_0&0
\end{pmatrix}.
\end{equation}
The admissible presentation and ordering parameters are
\begin{equation}\label{eq:Pbase}
\mathcal P_{\mathrm{base}}
=\mathbb R\times
\left(\mathbb R\setminus\left\{\frac{\hbar_0}{\vartheta_0B_0}\right\}\right)
\times\mathbb R,
\qquad (\varrho,r,s)\in\mathcal P_{\mathrm{base}}.
\end{equation}

\subsection{The sector-selected plane and its joint Weyl realization}

Let \(V=\mathfrak g_{\mathrm{NC}}/\mathfrak z\), where
\(\mathfrak z=\operatorname{span}\{Z_\hbar,Z_\vartheta,Z_B\}\).
Use the labelled central cocycles \(\xi,\xi',\xi''\) from
Section~\ref{subsec:central-character-fibres}, with parameter ordering
\((\Pi_x,\Pi_y,X,Y)\).  Their radicals give
\[
V_\Pi:=\operatorname{rad}\xi'=\operatorname{span}\{e_1,e_2\},
\qquad
V_X:=\operatorname{rad}\xi''=\operatorname{span}\{e_3,e_4\},
\qquad V=V_\Pi\oplus V_X.
\]
The lifts of these quotient spaces give a vector-space splitting of the Lie
algebra; their brackets take values in the centre.  The splitting of \(V\)
is intrinsic to the labelled extension: an automorphism preserving each
central line satisfies \(S^*\xi'=c_\vartheta\xi'\) and
\(S^*\xi''=c_B\xi''\), with nonzero scalars, and therefore preserves both
radicals.  No such assertion is made for automorphisms mixing the central
directions.  The nondegenerate fixed-sector form alone does not distinguish
these two planes.

The differential calculus of Section~\ref{sec:dirac-spectral-triple} is generated by
the sector momenta \(\Pi_x,\Pi_y\), and therefore singles out precisely the
labelled momentum plane \(V_\Pi\).  This is the reason for restricting the
inverse sector multiplier form to \(V_\Pi^*\) in the construction below.

The sector multiplier form has block matrix
\begin{equation}\label{eq:sector-omega-block}
\Omega_{\lambda_0}=2\pi\tau_{\hbar_0,\vartheta_0,B_0}
=\frac1{\hbar_0}
\begin{pmatrix}
B_0J_2&-I_2\\ I_2&(\vartheta_0/\hbar_0)J_2
\end{pmatrix},
\qquad J_2=\begin{pmatrix}0&1\\-1&0\end{pmatrix}.
\end{equation}
\begin{proposition}[The sector-selected bivector]
\label{prop:IV1}
Let \(\iota_\Pi:V_\Pi^*\hookrightarrow V^*\) be the inclusion determined
by the labelled splitting.  The bivector
\begin{equation}\label{eq:sector-restricted-bivector}
\Theta_{\lambda_0,\Pi}(\alpha,\beta)
:=\Omega_{\lambda_0}^{-1}(\iota_\Pi\alpha,\iota_\Pi\beta)
\end{equation}
has matrix
\begin{equation}\label{eq:sector-thetaeff-block}
[\Theta_{\lambda_0,\Pi}]_{(e^1,e^2)}
=\vartheta_{\mathrm{eff}}J_2.
\end{equation}
\end{proposition}

\begin{proof}
Block multiplication verifies
\begin{equation}\label{eq:sector-omega-inverse}
\Omega_{\lambda_0}^{-1}
=\frac1{1-\vartheta_0B_0/\hbar_0}
\begin{pmatrix}
\vartheta_0J_2&\hbar_0I_2\\
-\hbar_0I_2&\hbar_0B_0J_2
\end{pmatrix}.
\end{equation}
Its upper-left block proves \eqref{eq:sector-thetaeff-block}.
\end{proof}

\begin{remark}[Meaning and covariance of the selected scale]
\label{rem:thetaeff-covariance}
The compact-operator completion alone does not select a deformation
scale.  Here the labelled sector selects the bivector
\(\Theta_{\lambda_0,\Pi}\); its coefficient in the chosen basis is
\(\vartheta_{\mathrm{eff}}\).  The bivector is the invariant object.
For example, the change of units
\((X,Y)\mapsto\kappa(X,Y)\),
\((\Pi_x,\Pi_y)\mapsto\kappa^{-1}(\Pi_x,\Pi_y)\), \(\kappa>0\),
replaces \((\vartheta_0,B_0)\) by
\((\kappa^2\vartheta_0,\kappa^{-2}B_0)\).  Both
\(\vartheta_{\mathrm{eff}}\) and \(\vartheta_0\) acquire the same
factor \(\kappa^2\), and their ratio is unchanged.
To obtain a represented geometry, we next transport this effective plane
and the Dirac frame together.
\end{remark}

\begin{proposition}[Joint sector origin of the Moyal and Dirac frames]
\label{prop:sector-left-regular}
Define
\begin{equation}\label{eq:Qeff-sector}
Q_x^{\mathrm{eff}}=\frac{\hbar_0X+\vartheta_0\Pi_y}{\Delta},
\qquad
Q_y^{\mathrm{eff}}=\frac{\hbar_0Y-\vartheta_0\Pi_x}{\Delta}.
\end{equation}
Then
\begin{equation}\label{eq:Qeff-commutator}
[Q_x^{\mathrm{eff}},Q_y^{\mathrm{eff}}]
=-i\vartheta_{\mathrm{eff}}I.
\end{equation}
Choose the effective orientation
\(\widehat Q_x^{\mathrm{eff}}=Q_x^{\mathrm{eff}}\),
\(\widehat Q_y^{\mathrm{eff}}=-Q_y^{\mathrm{eff}}\), and set
\begin{equation}\label{eq:sector-complementary-pair}
C_x=X,\qquad C_y=-\frac{\hbar_0}{\Delta}Y.
\end{equation}
The complete frame satisfies
\begin{equation}\label{eq:sector-left-right-commutation-frame}
[\widehat Q_x^{\mathrm{eff}},\widehat Q_y^{\mathrm{eff}}]
=i\vartheta_{\mathrm{eff}}I,\qquad
[C_x,C_y]=-i\vartheta_{\mathrm{eff}}I,\qquad
[\widehat Q_j^{\mathrm{eff}},C_k]=0,
\end{equation}
where \(j,k\in\{x,y\}\) in the last relation.

For \(a^{\mathrm W}\in
\mathcal A_{\vartheta_{\mathrm{eff}},1/2}\), define the effective-coordinate
representation by Weyl integration of the sector pair,
\begin{equation}\label{eq:pi-eff-definition}
\pi_{\mathrm{eff}}^{r,s}(a^{\mathrm W})
:=\int_{\mathbb R^2}\widehat{a^{\mathrm W}}(k)\,
\exp\!\left(i\bigl(k_1\widehat Q_x^{\mathrm{eff}}
+k_2\widehat Q_y^{\mathrm{eff}}\bigr)\right)\,dk,
\end{equation}
with the Fourier convention used in Section~\ref{sec:affine-left-right-general};
the integral is understood in the usual Weyl-integrated sense.  Thus
\(\pi_{\mathrm{eff}}^{r,s}\) is the representation generated by the
sector-derived Heisenberg pair
\((\widehat Q_x^{\mathrm{eff}},\widehat Q_y^{\mathrm{eff}})\).
There is a metaplectic implementer \(\mathfrak M_{r,s}\) with
\begin{equation}\label{eq:joint-sector-moyal-frame}
\mathfrak M_{r,s}\widehat Q_j^{\mathrm{eff}}\mathfrak M_{r,s}^{-1}
=L_j^{(1/2)},\qquad
\mathfrak M_{r,s}C_j\mathfrak M_{r,s}^{-1}=R_j^{(1/2)},
\end{equation}
for \(j\in\{x,y\}\).
Consequently the Weyl-integrated effective-coordinate representation obeys
\begin{equation}\label{eq:sector-left-regular-equivalence}
\mathfrak M_{r,s}\pi_{\mathrm{eff}}^{r,s}(a^{\mathrm W})
 \mathfrak M_{r,s}^{-1}=L_{a^{\mathrm W}}^{(1/2)}.
\end{equation}
The same implementer transports the normalized sector momenta and hence the
Dirac operator.
\end{proposition}

\begin{proof}
To relate the effective coordinates directly to the inverse sector form,
use the parameter ordering of \(\Omega_{\lambda_0}\) and set
\[
\mathcal O:=\hbar_0^{-1}(\Pi_x,\Pi_y,X,Y)^T,
\qquad
\mathcal Q:=\Omega_{\lambda_0}^{-1}\mathcal O.
\]
The sector commutation relations give
\([\mathcal O_i,\mathcal O_j]
=i(\Omega_{\lambda_0})_{ij}I\). The displayed inverse matrix then yields
\[
\begin{pmatrix}\mathcal Q_1\\\mathcal Q_2\end{pmatrix}
=\frac1\Delta
\begin{pmatrix}
\hbar_0X+\vartheta_0\Pi_y\\
\hbar_0Y-\vartheta_0\Pi_x
\end{pmatrix}
=\begin{pmatrix}Q_x^{\mathrm{eff}}\\Q_y^{\mathrm{eff}}\end{pmatrix}.
\]
Since \(\Omega_{\lambda_0}^{-1}\) is antisymmetric, on the common
Schwartz core one has
\[
[\mathcal Q_i,\mathcal O_j]=i\delta_{ij}I,
\qquad
[\mathcal Q_i,\mathcal Q_j]
=-i(\Omega_{\lambda_0}^{-1})_{ij}I.
\]
Thus \(Q_x^{\mathrm{eff}},Q_y^{\mathrm{eff}}\) commute with \(X,Y\),
and the upper-left block of the inverse form gives
\[
[Q_x^{\mathrm{eff}},Q_y^{\mathrm{eff}}]
=-i\vartheta_{\mathrm{eff}}I.
\]
This identifies the operators as the momentum-dual components determined
by the inverse sector form and explains the sign before the chosen
orientation reversal. The definitions of
\(\widehat Q_x^{\mathrm{eff}},\widehat Q_y^{\mathrm{eff}},C_x,C_y\)
then give all the stated left--right commutation relations. Relative to
\(Z\), the frame determinant is
\begin{equation}\label{eq:sector-frame-det}
\det S_{\mathrm{sec}}=\frac{\hbar_0\vartheta_0^2}{\Delta^3}\ne0.
\end{equation}
Both this frame and the left--right Moyal frame are invertible linear
images of a canonical tuple, with the same commutation matrix.  Their
intertwining transformation is therefore symplectic and has a metaplectic
implementer.  Irreducibility of the full four-generator Weyl system makes
the intertwiner unique up to scalar phase.  Weyl integration gives
\eqref{eq:sector-left-regular-equivalence}; the momentum assertion follows
because the momenta are linear combinations of the same complete frame.
\end{proof}

\begin{remark}[Planar covariance and metric normalization]
\label{rem:metric-provenance}
The alternating sector data do not choose a positive momentum metric.
The induced stabilizer fixing all three labelled cocycles is
\[
\{g\oplus g^{-\mathsf T}:g\in SL(2,\mathbb R)\}:
\]
the pure momentum form forces \(\det g=1\), and the mixed form forces
the second block to be \(g^{-\mathsf T}\).  Every positive definite
symmetric \(2\times2\) matrix
\(M\) is \(M=\lambda^2g^{\mathsf T}g\), with
\(\lambda=(\det M)^{1/4}\) and \(g=\lambda^{-1}M^{1/2}\).
Thus positive momentum metrics are equivalent up to labelled symplectic
transport and one overall scale.  The isotropic representative used here is
\(\widetilde\Pi_x=\Pi_x/c_0\),
\(\widetilde\Pi_y=\varepsilon_0\Pi_y/c_0\), for which the commutator is
\(iI\).

Such a transport acts on the representation and Dirac operator together.
At general ordering the induced algebra map is
\(T_\varrho^{-1}\alpha_g^{(1/2)}T_\varrho\), where
\(\alpha_g^{(1/2)}\) is the induced area-preserving Weyl--Moyal
automorphism.  An external potential must be transported as well; it need
not remain in the two-parameter ansatz used below.  A further positive
rescaling \(D\mapsto\mu D\) scales distances by \(\mu^{-1}\), spectral
values by \(\mu\), and the two-dimensional volume by \(\mu^{-2}\), while
preserving the \(KK\)-class and local compactness.
\end{remark}

\begin{remark}[Position algebra versus effective coordinate algebra]
\label{rem:position-vs-effective}
The effective coordinates commute with \(X,Y\).  Under
\(\mathfrak M_{r,s}\), the effective plane gives the left action and
the rescaled position pair \((X,-\hbar_0Y/\Delta)\) gives the right
action.  The effective plane is the represented algebra chosen here;
local compactness does not characterize that choice uniquely.
Indeed, the alternative frame
\[
L^{\mathrm{pos}}=(X,Y)^{\mathsf T},\qquad
R^{\mathrm{pos}}=
\left(\widehat Q_x^{\mathrm{eff}},
 -\frac{\Delta}{\hbar_0}\widehat Q_y^{\mathrm{eff}}\right)^{\mathsf T}
\]
has left and right commutators \(i\vartheta_0I\) and
\(-i\vartheta_0I\).  Since
\[
\Pi_x=\frac{\hbar_0}{\vartheta_0}
 (L^{\mathrm{pos}}_2-R^{\mathrm{pos}}_2),\qquad
\Pi_y=\frac{\Delta}{\vartheta_0}R^{\mathrm{pos}}_1
 -\frac{\hbar_0}{\vartheta_0}L^{\mathrm{pos}}_1,
\]
its normalized Dirac data have
\(\Gamma_{\mathrm{pos}}=1/2\) and
\(\det T_R^{\mathrm{pos}}
=\varepsilon_0\hbar_0\Delta/(2c_0^2\vartheta_0^2)\ne0\).
The affine theorem gives a locally compact triple with Gaussian volume
\(2|\vartheta_0B_0/\Delta|\).  Affine left gauge potentials for this
alternative representation are functions of \(X,Y\).
\end{remark}

\begin{remark}[Ordering transport and parameter roles]
\label{rem:ordering-change}
Use the ordering transport \(T_\varrho\) of
Section~\ref{sec:affine-left-right-general} with
\(\vartheta=\vartheta_{\mathrm{eff}}\).  Thus
\(a^{\mathrm W}=T_\varrho a_\varrho\), and
\begin{equation}\label{eq:ordering-intertwines-products}
T_\varrho L_{a_\varrho}^{(\varrho)}T_\varrho^{-1}
=L_{a^{\mathrm W}}^{(1/2)},\qquad
T_\varrho R_{a_\varrho}^{(\varrho)}T_\varrho^{-1}
=R_{a^{\mathrm W}}^{(1/2)}.
\end{equation}
Set \(W_{\varrho,r,s}=T_\varrho^{-1}\mathfrak M_{r,s}\) and
\(\pi_{\mathrm{eff},\varrho}^{r,s}(a_\varrho)
=\pi_{\mathrm{eff}}^{r,s}(T_\varrho a_\varrho)\).  Since \(T_\varrho\)
fixes affine functions,
\begin{equation}\label{eq:ordering-left-right}
W_{\varrho,r,s}\widehat Q_j^{\mathrm{eff}}W_{\varrho,r,s}^{-1}
=L_j^{(\varrho)},\qquad
W_{\varrho,r,s}C_jW_{\varrho,r,s}^{-1}=R_j^{(\varrho)},
\end{equation}
for \(j\in\{x,y\}\),
and
\begin{equation}\label{eq:sector-left-regular-general-varrho}
W_{\varrho,r,s}\pi_{\mathrm{eff},\varrho}^{r,s}(a_\varrho)
 W_{\varrho,r,s}^{-1}=L_{a_\varrho}^{(\varrho)}.
\end{equation}
Thus \((r,s)\) changes the sector presentation and \(\varrho\) changes
the ordering realization; both are implemented by joint transport of the
algebra and Dirac operator.  Changing an external gauge potential is a
separate operation.
\end{remark}

\subsection{Jointly transported Dirac operator and base spectral triple}
\label{subsec:dirac-axioms}

Use \(\mathcal A_{\vartheta_{\mathrm{eff}},\varrho}
=(\mathcal S(\mathbb R^2),\star_{\vartheta_{\mathrm{eff}},\varrho},
{}^{*_{\varrho}})\), represented by
\begin{equation}\label{eq:pi2-def}
\pi_{2,\varrho}(a_\varrho)=L_{a_\varrho}^{(\varrho)},\qquad
\pi^M_{2,\varrho}(a_\varrho)=L_{a_\varrho}^{(\varrho)}\otimes I_2.
\end{equation}
On \(\mathcal H=L^2(\mathbb R^2)\otimes\mathbb C^2\), put
\begin{equation}\label{eq:W-varrho-rs}
\mathcal W_{\varrho,r,s}=W_{\varrho,r,s}\otimes I_2,
\end{equation}
and transport the sector operator and its domain by
\begin{equation}\label{eq:Deff-def}
D_{\mathrm{eff}}^{\varrho,r,s}
=\mathcal W_{\varrho,r,s}D^{r,s}_{\hbar_0,\vartheta_0,B_0}
 \mathcal W_{\varrho,r,s}^{-1}.
\end{equation}
The transport preserves Schwartz spinors.  Solving
\eqref{eq:Qeff-sector} gives
\begin{equation}\label{eq:Pi-in-QC-frame}
\Pi_x=\frac{\Delta}{\vartheta_0}\widehat Q_y^{\mathrm{eff}}
       +\frac{\hbar_0}{\vartheta_0}Y,\qquad
\Pi_y=\frac{\Delta}{\vartheta_0}\widehat Q_x^{\mathrm{eff}}
       -\frac{\hbar_0}{\vartheta_0}X.
\end{equation}
Hence the transported normalized momenta are
\begin{equation}\label{eq:Keff-left-right}
\begin{aligned}
\mathsf K_{x,\varrho}
&=\frac{\Delta}{\vartheta_0c_0}
 (L_y^{(\varrho)}-R_y^{(\varrho)}),\\
\mathsf K_{y,\varrho}
&=\frac{\varepsilon_0}{\vartheta_0c_0}
 (\Delta L_x^{(\varrho)}-\hbar_0R_x^{(\varrho)}).
\end{aligned}
\end{equation}
Using the canonical formulas \eqref{eq:general-left-right-symbols} with
\(\vartheta=\vartheta_{\mathrm{eff}}\), where \(X_j\) denotes multiplication
by \(x_j\) and \(P_j=-i\partial_{x_j}\),
\eqref{eq:Keff-left-right} becomes
\begin{equation}\label{eq:Keff-canonical-form}
\mathsf K_{x,\varrho}=\frac{\hbar_0}{c_0}P_1,\qquad
\mathsf K_{y,\varrho}=-\frac{\varepsilon_0B_0}{c_0}X_1
-\frac{\varepsilon_0\hbar_0}{c_0}
 \left(1+\frac{\varrho\vartheta_0B_0}{\Delta}\right)P_2,
\end{equation}
and
\begin{equation}\label{eq:Keff-CCR}
[\mathsf K_{x,\varrho},\mathsf K_{y,\varrho}]=iI.
\end{equation}
The Dirac operator is
\begin{equation}\label{eq:Deff-normal-form}
D_{\mathrm{eff}}^{\varrho,r,s}
=\frac1{\sqrt2}\left(
\mathsf K_{x,\varrho}\otimes\sigma_1+
\mathsf K_{y,\varrho}\otimes\sigma_2\right).
\end{equation}
The explicit normal form has no \((r,s)\)-dependence, because the
presentation dependence is absorbed into \(\mathcal W_{\varrho,r,s}\).
Its affine blocks are
\begin{equation}\label{eq:NCQM-TL0-TR}
T_L^0=\frac1{\sqrt2c_0}
\begin{pmatrix}0&\Delta/\vartheta_0\\
\varepsilon_0\Delta/\vartheta_0&0\end{pmatrix},\qquad
T_R=\frac1{\sqrt2c_0}
\begin{pmatrix}0&-\Delta/\vartheta_0\\
-\varepsilon_0\hbar_0/\vartheta_0&0\end{pmatrix}.
\end{equation}
In particular,
\begin{equation}\label{eq:NCQM-general-theorem-checks}
\vartheta_{\mathrm{eff}}^2\det T_R
=-\frac{\varepsilon_0\hbar_0^3}{2\Delta c_0^2}\ne0,
\qquad
\vartheta_{\mathrm{eff}}(\det T_L^0-\det T_R)=\frac12.
\end{equation}

\begin{theorem}[The transported base triple]
\label{thm:moyal-base-triple}
For every \((\varrho,r,s)\in\mathcal P_{\mathrm{base}}\), the data
\begin{equation}\label{eq:joint-sector-base-triple}
\left(\mathcal A_{\vartheta_{\mathrm{eff}},\varrho},\mathcal H,
\pi^M_{2,\varrho},D_{\mathrm{eff}}^{\varrho,r,s},I\otimes\sigma_3\right)
\end{equation}
form an even locally compact non-unital spectral triple.  Its algebra representation and Dirac operator are obtained by joint transport
of the effective-coordinate representation and the sector Dirac data.  In particular,
\begin{equation}\label{eq:Deff-commutator-final}
[D_{\mathrm{eff}}^{\varrho,r,s},\pi^M_{2,\varrho}(a_\varrho)]
=-\frac{i\hbar_0}{\sqrt2c_0}
 L_{\partial_xa_\varrho}^{(\varrho)}\otimes\sigma_1
+\frac{i\varepsilon_0\hbar_0}{\sqrt2c_0}
 L_{\partial_ya_\varrho}^{(\varrho)}\otimes\sigma_2.
\end{equation}
\end{theorem}

\begin{proof}
The joint transport follows from
\eqref{eq:sector-left-regular-general-varrho} and \eqref{eq:Deff-def}.
Using the affine Moyal commutators in \eqref{eq:Keff-left-right} gives
\begin{equation}\label{eq:Keff-left-commutators}
[\mathsf K_{x,\varrho},L_{a_\varrho}^{(\varrho)}]
=-\frac{i\hbar_0}{c_0}L_{\partial_xa_\varrho}^{(\varrho)},\qquad
[\mathsf K_{y,\varrho},L_{a_\varrho}^{(\varrho)}]
=\frac{i\varepsilon_0\hbar_0}{c_0}L_{\partial_ya_\varrho}^{(\varrho)},
\end{equation}
which proves the displayed formula.  The spectral-triple assertions,
including domain invariance, follow from the affine theorem and
\(\det T_R\ne0\).
\end{proof}

For the symbol calculations below, write
\(z=(x_1,x_2,\xi_1,\xi_2)\) and recall the left-coordinate symbols
\[
\nu_1=x_1-(1-\varrho)\vartheta_{\mathrm{eff}}\xi_2,
\qquad
\nu_2=x_2+\varrho\vartheta_{\mathrm{eff}}\xi_1.
\]
The Weyl symbol of left multiplication is
\[
L_{a_\varrho}^{(\varrho)}
=\operatorname{Op}^{\mathrm W}\!\bigl(a^{\mathrm W}\circ\nu\bigr),
\qquad a^{\mathrm W}=T_\varrho a_\varrho.
\]
Let \(q_0,p_0\) be the Weyl symbols of
\(\mathsf K_{x,\varrho}/\sqrt2\) and
\(\mathsf K_{y,\varrho}/\sqrt2\), respectively.  Thus these are the
Dirac-component symbols, including the Clifford normalization in
\eqref{eq:Deff-normal-form}; their joint Jacobian with \(\nu\) is
\(\vartheta_{\mathrm{eff}}^2\det T_R\), as in
\eqref{eq:NCQM-general-theorem-checks}.

\subsection{Affine gauge coupling and transversality}
\label{subsec:localized-gauge-perturbations}

Fix \(e,\alpha_x,\alpha_y\in\mathbb R\) and the affine ansatz
\begin{equation}\label{eq:general-affine-gauge-pair}
A_x=\alpha_x y,\qquad A_y=\alpha_y x.
\end{equation}
With the curvature convention
\begin{equation}\label{eq:affine-curvature-def}
F_{xy}=\partial_xA_y-\partial_yA_x
-\frac{ie}{\hbar_0}[A_x,A_y]_{\star_{\vartheta_{\mathrm{eff}},\varrho}},
\end{equation}
the exact affine commutator gives
\begin{equation}\label{eq:affine-curvature-general}
\boxed{F_{xy}=\alpha_y-\alpha_x
-\frac{e\vartheta_{\mathrm{eff}}}{\hbar_0}\alpha_x\alpha_y.}
\end{equation}
The Landau representative \(A_x=-F_{xy}y\), \(A_y=0\) realizes every
real value of the curvature.

The affine multiplier coordinates
\begin{equation}\label{eq:bx-by-def}
b_x=\frac{i}{\hbar_0}x,\qquad b_y=\frac{i}{\hbar_0}y
\end{equation}
have the bounded commutators
\begin{align}
[D_{\mathrm{eff}},\pi^M_{2,\varrho}(b_x)]
&=\frac1{\sqrt2c_0}I\otimes\sigma_1,
\label{eq:coord-comm-x}\\
[D_{\mathrm{eff}},\pi^M_{2,\varrho}(b_y)]
&=-\frac{\varepsilon_0}{\sqrt2c_0}I\otimes\sigma_2,
\label{eq:coord-comm-y}
\end{align}
by \eqref{eq:Deff-commutator-final}.  Here \(\pi^M_{2,\varrho}\) is
extended to affine multipliers on the Schwartz core.

On \(\mathcal S(\mathbb R^2)\otimes\mathbb C^2\), define the formal
multiplier one-form
\begin{equation}\label{eq:Ainf-formal}
\mathcal A_D=\sum_{j=x,y}\pi^M_{2,\varrho}(A_j)
 [D_{\mathrm{eff}},\pi^M_{2,\varrho}(b_j)]
\end{equation}
and the formal minimally coupled operator
\begin{equation}\label{eq:Dinf-formal}
D_\infty^\circ=D_{\mathrm{eff}}^{\varrho,r,s}-e\mathcal A_D,
\qquad \operatorname{Dom}D_\infty^\circ
=\mathcal S(\mathbb R^2)\otimes\mathbb C^2.
\end{equation}
Explicitly,
\begin{equation}\label{eq:Ainf-explicit}
\mathcal A_D=\frac1{\sqrt2c_0}
 \left(L_{A_x}^{(\varrho)}\otimes\sigma_1
 -\varepsilon_0L_{A_y}^{(\varrho)}\otimes\sigma_2\right),
\end{equation}
so
\begin{equation}\label{eq:Dinf-explicit}
D_\infty^\circ=D_{\mathrm{eff}}^{\varrho,r,s}
-\frac e{\sqrt2c_0}
 \left(L_{A_x}^{(\varrho)}\otimes\sigma_1
 -\varepsilon_0L_{A_y}^{(\varrho)}\otimes\sigma_2\right).
\end{equation}
The charge \(e\) therefore multiplies the potentials in the same
normalization as the sector momenta.  The affine multipliers
\begin{align}
L_{A_x}^{(\varrho)}
&=\alpha_x(X_2+\varrho\vartheta_{\mathrm{eff}}P_1),
\label{eq:LAx-explicit-proof}\\
L_{A_y}^{(\varrho)}
&=\alpha_y(X_1-(1-\varrho)\vartheta_{\mathrm{eff}}P_2)
\label{eq:LAy-explicit-proof}
\end{align}
are symmetric on Schwartz space.  In the left--right frame the coupled
blocks are
\begin{equation}\label{eq:NCQM-gauge-TL-TR}
T_L(\alpha_x,\alpha_y)=T_L^0+
\frac e{\sqrt2c_0}
\begin{pmatrix}0&-\alpha_x\\\varepsilon_0\alpha_y&0\end{pmatrix},
\qquad T_R\text{ as in \eqref{eq:NCQM-TL0-TR}}.
\end{equation}
Thus the right block is unchanged by the coupling.

For comparison with canonical coordinates, write
\begin{equation}\label{eq:Dinf-QP-form}
D_\infty^\circ=\mathsf Q_B\otimes\sigma_1+\mathsf P_B\otimes\sigma_2,
\end{equation}
where
\begin{align}
\mathsf Q_B&=\frac1{\sqrt2c_0}
 \left[(\hbar_0-e\varrho\vartheta_{\mathrm{eff}}\alpha_x)P_1
 -e\alpha_xX_2\right],\label{eq:QB-def}\\
\mathsf P_B&=-\frac{\varepsilon_0}{\sqrt2c_0}
 \left[(B_0-e\alpha_y)X_1+
 \left(\hbar_0\left(1+\frac{\varrho\vartheta_0B_0}{\Delta}\right)
 +e(1-\varrho)\vartheta_{\mathrm{eff}}\alpha_y\right)P_2\right].
\label{eq:PB-def}
\end{align}
Let \(q_B,p_B\) be their Weyl symbols, and put
\begin{equation}\label{eq:gauge-linear-map}
\mathscr L_B:z\longmapsto(\nu_1,\nu_2,q_B,p_B),\qquad
A_B:=d\mathscr L_B
=\tfrac{\partial(\nu_1,\nu_2,q_B,p_B)}{\partial(x_1,x_2,\xi_1,\xi_2)}.
\end{equation}

\begin{theorem}[Affine coupling and transversality]
\label{thm:gauge-transversality}
For every \((\varrho,r,s)\in\mathcal P_{\mathrm{base}}\) and every
real \(\alpha_x,\alpha_y\), \(D_\infty^\circ\) is essentially
self-adjoint on Schwartz spinors.  Its closure \(D_\infty\) defines the
even locally compact spectral triple
\begin{equation}\label{eq:Dinf-triple}
\left(\mathcal A_{\vartheta_{\mathrm{eff}},\varrho},\mathcal H,
\pi^M_{2,\varrho},D_\infty,I\otimes\sigma_3\right).
\end{equation}
Minimal coupling is the shear
\begin{equation}\label{eq:gauge-shear}
\begin{pmatrix}q_B\\p_B\end{pmatrix}
=\begin{pmatrix}q_0\\p_0\end{pmatrix}
+\frac e{\sqrt2c_0}
\begin{pmatrix}0&-\alpha_x\\\varepsilon_0\alpha_y&0\end{pmatrix}
\begin{pmatrix}\nu_1\\\nu_2\end{pmatrix}.
\end{equation}
It changes only the left block in \eqref{eq:NCQM-gauge-TL-TR}, so the
Jacobian determinant of \(\mathscr L_B\) is
\begin{equation}\label{eq:gauge-transversality-det}
\boxed{\det A_B
=\vartheta_{\mathrm{eff}}^2\det T_R
=-\frac{\varepsilon_0\hbar_0^3}{2\Delta c_0^2}\ne0.}
\end{equation}
It is independent of \(e,\alpha_x,\alpha_y,\varrho,r,s\).
In particular \(q_B,p_B\) remain linearly independent when their
Poisson bracket vanishes.
\end{theorem}

\begin{proof}
Equations \eqref{eq:Dinf-explicit} and \eqref{eq:NCQM-gauge-TL-TR}
identify the operator with the affine data of
Section~\ref{sec:affine-left-right-general}.  The affine argument gives essential self-adjointness on Schwartz spinors,
which form a core for \(D_\infty\).
Theorem~\ref{thm:affine-left-right}(ii) gives the spectral triple because
the right block is invertible.  Equation \eqref{eq:gauge-shear} follows
from the linear dependence of the potentials on the left coordinates.
As a transformation of \((\nu_1,\nu_2,q_0,p_0)\), this shear has a
block-triangular Jacobian with both diagonal blocks equal to \(I_2\).
Its determinant is therefore one, so \(\det A_B\) remains equal to the uncoupled value in \eqref{eq:NCQM-general-theorem-checks}.
The nonvanishing of \(\det A_B\) also proves the linear independence of \(q_B,p_B\).
\end{proof}

The coupled bounded commutator is explicitly
\begin{equation}\label{eq:Dinf-bounded-commutator}
\begin{aligned}
[D_\infty,\pi^M_{2,\varrho}(a_\varrho)]
={}&-\frac{i(\hbar_0-e\vartheta_{\mathrm{eff}}\alpha_x)}{\sqrt2c_0}
 L_{\partial_xa_\varrho}^{(\varrho)}\otimes\sigma_1\\
&+\frac{i\varepsilon_0(\hbar_0+e\vartheta_{\mathrm{eff}}\alpha_y)}
 {\sqrt2c_0}L_{\partial_ya_\varrho}^{(\varrho)}\otimes\sigma_2.
\end{aligned}
\end{equation}
\begin{remark}[Gauge covariance]
\label{rem:explicit-gauge-covariance}
Let \(e\ne0\), and let \(v\) be a smooth star-unitary Moyal multiplier
such that \(L_v^{(\varrho)}\) and \(L_{v^{*_{\varrho}}}^{(\varrho)}\)
preserve Schwartz space.  Suppress the fixed ordering subscripts in
products and set
\begin{equation}\label{eq:gauge-transformation-potentials}
A_j^v=v\star A_j\star v^*
       +\frac{i\hbar_0}{e}v\star\partial_jv^*,\qquad j=x,y.
\end{equation}
Writing \(D(A)\) for the closure of \eqref{eq:Dinf-explicit},
\eqref{eq:Deff-commutator-final} and the Leibniz rule give, on the
Schwartz core,
\begin{equation}\label{eq:gauge-transformation-dirac-curvature}
\begin{aligned}
D(A^v)&=(L_v^{(\varrho)}\otimes I_2)D(A)
        (L_v^{(\varrho)}\otimes I_2)^*,\\
F_{xy}^v&=v\star F_{xy}\star v^*.
\end{aligned}
\end{equation}
The first equality specifies the self-adjoint closure of the transformed
operator.  Scalar affine curvature is therefore gauge invariant.
A general transformation need not preserve the affine ansatz; the exact
affine formulas below apply when the potential remains affine.
\end{remark}

\subsection{Spectral response, volume, and the even class}

\begin{theorem}[Curvature-dependent spectral normal form]
\label{thm:Dinf-spectrum}
The affine generators obey
\begin{equation}\label{eq:Gamma-def}
[\mathsf Q_B,\mathsf P_B]=i\Gamma(F_{xy})I,
\end{equation}
where
\begin{equation}\label{eq:Gamma-formula}
\boxed{\Gamma(F_{xy})
=\frac12\left(1-\frac{eF_{xy}}{B_0}\right).}
\end{equation}
If \(eF_{xy}\ne B_0\), then
\begin{equation}\label{eq:Dinf-spectrum-noncritical}
\operatorname{Spec}(D_\infty)
=\{0\}\cup\left\{\pm\sqrt{
 \left|1-\frac{eF_{xy}}{B_0}\right|n}:n\ge1\right\},
\end{equation}
and every eigenvalue has infinite multiplicity.  On the cancellation locus
\begin{equation}\label{eq:critical-locus}
\boxed{eF_{xy}=B_0,}
\end{equation}
one instead has
\begin{equation}\label{eq:Dinf-spectrum-critical}
\operatorname{Spec}(D_\infty)=\mathbb R,
\qquad \sigma_{\mathrm p}(D_\infty)=\varnothing.
\end{equation}
For \(e\ne0\), cancellation occurs at \(F_{xy}=B_0/e\); for \(e=0\),
no critical member occurs.
\end{theorem}

\begin{proof}
The determinant formula of Theorem~\ref{thm:affine-left-right}(i),
applied to \eqref{eq:NCQM-gauge-TL-TR}, gives
\[
\begin{aligned}
\Gamma
&=\vartheta_{\mathrm{eff}}(\det T_L(\alpha_x,\alpha_y)-\det T_R)\\
&=\frac12-\frac e{2B_0}(\alpha_y-\alpha_x)
 +\frac{e^2\vartheta_{\mathrm{eff}}}{2\hbar_0B_0}\alpha_x\alpha_y
=\frac12-\frac e{2B_0}F_{xy}.
\end{aligned}
\]
The spectral assertions are then exactly the two cases of
Theorem~\ref{thm:affine-left-right}(v), since \(T_R\) remains invertible.
\end{proof}

\begin{corollary}[Unitary classification of the coupled operators]
\label{cor:curvature-spectral-response}
Two noncritical members are unitarily equivalent as ungraded self-adjoint
operators if and only if they have the same value of
\(\left|1-eF_{xy}/B_0\right|\).
All critical members are unitarily equivalent, and none is equivalent to
a noncritical member.  This statement does not classify the represented
graded spectral triples.
\end{corollary}

\begin{proof}
For noncritical members, the smallest positive squared eigenvalue recovers
the displayed parameter.  Conversely, equal parameters give the same
pure point eigenvalues, each of countably infinite multiplicity, so an
eigenbasis identification gives a unitary equivalence.  All critical
members have the common normal form \eqref{eq:general-free-normal-form}
and no eigenvalues, which also separates them from the noncritical case.
\end{proof}

\begin{corollary}[Metric rank and spectral cancellation]
\label{cor:metric-vs-spectral-degeneracy}
For the two-parameter ansatz, put
\(u_x=\hbar_0-e\vartheta_{\mathrm{eff}}\alpha_x\) and
\(u_y=\hbar_0+e\vartheta_{\mathrm{eff}}\alpha_y\).
Then
\begin{equation}\label{eq:metric-coefficient-determinant}
u_xu_y=\hbar_0^2
 \left(1+\frac{e\vartheta_0F_{xy}}{\Delta}\right).
\end{equation}
The left block is singular precisely when
\begin{equation}\label{eq:metric-rank-locus}
eF_{xy}=-\frac{\Delta}{\vartheta_0}
       =B_0-\frac{\hbar_0}{\vartheta_0}.
\end{equation}
This locus is disjoint from spectral cancellation \(eF_{xy}=B_0\),
where \(\det T_L=\det T_R\ne0\).
For \(e\ne0\), the Lipschitz seminorm vanishes on the whole algebra
precisely at
\begin{equation}\label{eq:zero-metric-seminorm-point}
\alpha_x=\frac{\Delta}{e\vartheta_0},\qquad
\alpha_y=-\frac{\Delta}{e\vartheta_0}.
\end{equation}
Other points on \eqref{eq:metric-rank-locus} have rank-one left block.
\end{corollary}

\begin{proof}
Expand \(u_xu_y\) and use \eqref{eq:affine-curvature-general} and
\(\vartheta_{\mathrm{eff}}=\hbar_0\vartheta_0/\Delta\).
The two off-diagonal entries of \(T_L\) are nonzero fixed multiples
of \(u_x,u_y\), so its rank follows immediately.  The two Pauli
coefficients in \eqref{eq:Dinf-bounded-commutator} vanish for every
Schwartz element exactly when \(u_x=u_y=0\), by faithfulness of left
multiplication.  For \(e\ne0\), fixing
\(\alpha_x=\Delta/(e\vartheta_0)\) and varying \(\alpha_y\) gives
the same curvature with rank zero at the displayed point and rank one
elsewhere.  Thus scalar curvature does not determine the metric seminorm
on the fixed algebra.
\end{proof}

\begin{remark}[The critical free-Dirac normal form]
\label{rem:critical-free-dirac}
At cancellation the two independent Dirac symbols Poisson commute and
span a Lagrangian plane.  The metaplectic normal form of
Section~\ref{sec:affine-left-right-general}, followed by a Fourier
transform, identifies the operator with a free two-dimensional Dirac
operator.  This is an operator normal form; the algebra representation is
transported with it and need not be the standard left Moyal representation
in those new coordinates.  Local compactness survives because the
sector-fixed right block stays invertible.
\end{remark}

\subsubsection{Localized heat trace and noncommutative volume}
\label{subsubsec:heat-volume}

Write \(\Gamma=\Gamma(F_{xy})\).  Use the normalized trace
\(\tau_{\vartheta_{\mathrm{eff}}}\) and Gaussian generator \(p_\varrho\)
from Section~\ref{sec:affine-left-right-general}, now with
\(\vartheta=\vartheta_{\mathrm{eff}}\).  Thus
\(\tau_{\vartheta_{\mathrm{eff}}}(p_\varrho)=1\).

\begin{theorem}[Localized heat trace, spectral dimension, and volume]
\label{thm:exact-localized-heat-trace}
For every \(a_\varrho\in\mathcal A_{\vartheta_{\mathrm{eff}},\varrho}\)
and \(t>0\), the localized heat operator is trace class and
\begin{equation}\label{eq:exact-localized-heat-trace}
\operatorname{Tr}\!\left(
 \pi^M_{2,\varrho}(a_\varrho)e^{-tD_\infty^2}\right)
=\frac{\int a_\varrho}{2\pi|\det A_B|}\Phi_\Gamma(t),
\end{equation}
with \(\Phi_\Gamma\) as in \eqref{eq:general-affine-Phi}.
The local spectral dimension is two.  The zeta function extends across
\(s=2\) with at most a simple pole and residue
\begin{equation}\label{eq:localized-zeta-residue}
\operatorname*{Res}_{s=2}\operatorname{Tr}\!\left(
\pi^M_{2,\varrho}(a_\varrho)(1+D_\infty^2)^{-s/2}\right)
=\frac{\int a_\varrho}{\pi|\det A_B|}.
\end{equation}
Moreover, \(\pi^M_{2,\varrho}(a_\varrho)(1+D_\infty^2)^{-1}\) lies
in \(\mathcal Z_1\), is measurable in the convention of
Section~\ref{sec:affine-left-right-general}, and satisfies
\begin{equation}\label{eq:dixmier-volume}
\boxed{\operatorname{Tr}_\omega\!\left(
\pi^M_{2,\varrho}(a_\varrho)(1+D_\infty^2)^{-1}\right)
=2\left|\frac{\vartheta_0B_0}{\hbar_0}\right|
 \tau_{\vartheta_{\mathrm{eff}}}(a_\varrho).}
\end{equation}
These volume coefficients and the local dimension are independent of the
external field, including at cancellation.
\end{theorem}

\begin{proof}
Apply Proposition~\ref{prop:general-affine-analytic-package} with
\(\vartheta=\vartheta_{\mathrm{eff}}\) and the blocks
\eqref{eq:NCQM-gauge-TL-TR}.  The required Jacobian factor is
\(|\det A_B|=\vartheta_{\mathrm{eff}}^2|\det T_R|\).
The sector constants give
\[
\frac1{|\det A_B|}
=\frac{2|\Delta|\,|B_0|}{|\hbar_0|^2},\qquad
|\vartheta_{\mathrm{eff}}|=\frac{|\hbar_0\vartheta_0|}{|\Delta|},
\]
which proves the normalized volume coefficient.  Its value on the
Gaussian is positive, so the critical summability exponent is two.
\end{proof}

\begin{remark}[Ultraviolet coefficient and infrared transition]
\label{rem:heat-UV-IR}
The ultraviolet coefficient is field-independent, whereas
\(\Phi_\Gamma(t)\to|\Gamma|\) for \(\Gamma\ne0\) and
\(\Phi_0(t)=t^{-1}\to0\) as \(t\to\infty\).  For a nonzero positive
localization this distinguishes the zero Landau level from the critical
regime.  The small-time expansion is locally uniform for bounded
\(\Gamma\); its leading coefficient is independent of \(\Gamma\).
\end{remark}

\subsubsection{The even \texorpdfstring{\(K\)}{K}-homology class}

Let \(A_{\vartheta_{\mathrm{eff}},\varrho}\) be the \(C^*\)-completion
of the represented Moyal algebra.  It is isomorphic to the compact
operators, and the Gaussian class is its positive \(K_0\)-generator.

\begin{corollary}[Uniform even class]
\label{thm:gauge-index}
For every \((\alpha_x,\alpha_y)\in\mathbb R^2\), including those at
magnetic cancellation,
\begin{equation}\label{eq:gauge-index-pairing}
\boxed{\langle[p_\varrho],[D_\infty]\rangle
=-\operatorname{sgn}(\vartheta_{\mathrm{eff}})
\operatorname{sgn}(\det A_B)
=\varepsilon_0\operatorname{sgn}(\vartheta_0).}
\end{equation}
Consequently,
\begin{equation}\label{eq:affine-KK-equal}
[D_\infty]=[D_{\mathrm{eff}}^{\varrho,r,s}]
\quad\text{in }KK^0(A_{\vartheta_{\mathrm{eff}},\varrho},\mathbb C).
\end{equation}
In particular, these classes are generators with the displayed orientation.
\end{corollary}

\begin{proof}
Theorem~\ref{thm:affine-left-right}(iv) gives
\(-\operatorname{sgn}(\vartheta_{\mathrm{eff}})
\operatorname{sgn}(\det T_R)\).
Since \(\vartheta_{\mathrm{eff}}^2\det T_R=\det A_B\),
substituting \eqref{eq:gauge-transversality-det} gives
\(\varepsilon_0\operatorname{sgn}(\vartheta_0)\), both for
\(D_\infty\) and for the base operator.  Morita equivalence with
\(\mathbb C\) identifies pairing with \([p_\varrho]\) as the
isomorphism \(KK^0(A_{\vartheta_{\mathrm{eff}},\varrho},\mathbb C)
\to\mathbb Z\); hence the equal pairings give equal classes.
\end{proof}

\begin{remark}[Orientation convention]
\label{rem:index-orientation}
The sign \(\varepsilon_0\operatorname{sgn}(\vartheta_0)\) uses the
normalized momentum frame and the grading \(I\otimes\sigma_3\).
The computation of the pairing establishes that the class has magnitude
one; this does not follow merely from \(K_0\cong\mathbb Z\).
\end{remark}

\begin{corollary}[Arbitrary real affine potentials]
\label{cor:full-affine-potentials}
Let
\begin{equation}\label{eq:full-affine-potential}
\begin{pmatrix}A_x\\A_y\end{pmatrix}
=\mathsf A\begin{pmatrix}x\\y\end{pmatrix}+b,
\qquad \mathsf A=(a_{jk})\in M_2(\mathbb R),\quad b\in\mathbb R^2.
\end{equation}
Its scalar curvature is
\begin{equation}\label{eq:full-affine-curvature}
F_{xy}=a_{21}-a_{12}
       +\frac{e\vartheta_{\mathrm{eff}}}{\hbar_0}\det\mathsf A.
\end{equation}
The closure of \eqref{eq:Dinf-explicit} is an even locally compact
spectral triple, with affine data
\begin{equation}\label{eq:full-affine-blocks}
T_L=T_L^0+\frac e{\sqrt2c_0}E\mathsf A,\qquad
T_R\text{ as in \eqref{eq:NCQM-TL0-TR}},\qquad
c=\frac e{\sqrt2c_0}Eb,\quad E=\operatorname{diag}(-1,\varepsilon_0).
\end{equation}
Equations \eqref{eq:gauge-transversality-det} and
\eqref{eq:Gamma-formula} remain valid.  Consequently the spectral,
heat, volume and \(KK\)-class conclusions above hold for every real
affine potential, including arbitrary constant terms.
\end{corollary}

\begin{proof}
The affine commutator is
\([A_x,A_y]_\star=i\vartheta_{\mathrm{eff}}\det\mathsf A\), which
gives the curvature formula.  Reading the coefficients in
\eqref{eq:Dinf-explicit} gives \eqref{eq:full-affine-blocks}, and
expanding its determinant yields
\[
\vartheta_{\mathrm{eff}}(\det T_L-\det T_R)
=\frac12-\frac e{2B_0}
 \left(a_{21}-a_{12}
       +\frac{e\vartheta_{\mathrm{eff}}}{\hbar_0}\det\mathsf A\right).
\]
Apply Theorem~\ref{thm:affine-left-right}.  Since the right block is
unchanged and invertible, its Weyl displacement
\eqref{eq:general-c-removal} removes \(c\) while commuting with the
represented algebra.  The index pairing, and hence the even class,
is therefore unchanged.
\end{proof}

\subsection{Finite cutoffs and their removal}
\label{subsec:cutoff-removal}

The uncut affine triple has already been established directly.  The following
approximation records its additional gauge-theoretic interpretation as a
strong-resolvent limit of bounded inner fluctuations.

Choose a real \(u\in C_c^\infty(\mathbb R^2)\) with
\(0\le u\le1\), \(u=1\) on \(B_1(0)\), and
\(\operatorname{supp}u\subset B_2(0)\), and set
\(u_R(z)=u(z/R)\) for \(R>0\).
For any affine potential of Corollary~\ref{cor:full-affine-potentials}, define
\begin{equation}\label{eq:varrho-localized-potentials}
A_{j,\varrho}^{(R)}=T_\varrho^{-1}(u_RA_j),\qquad j=x,y.
\end{equation}
These elements are Schwartz and self-adjoint for \({}^{*_{\varrho}}\),
since their Weyl representatives \(u_RA_j\) are real Schwartz functions.
Set
\begin{equation}\label{eq:A-R-def}
\mathcal A_D^{(R)}
=\sum_{j=x,y}\pi^M_{2,\varrho}(A_{j,\varrho}^{(R)})
 [D_{\mathrm{eff}},\pi^M_{2,\varrho}(b_j)]
\end{equation}
and
\begin{equation}\label{eq:DR-def}
D_R^{\varrho,r,s}
=D_{\mathrm{eff}}^{\varrho,r,s}-e\mathcal A_D^{(R)}.
\end{equation}
The operator \(\mathcal A_D^{(R)}\) is a represented one-form over the
Schwartz algebra.  Indeed, Moyal factorization
\cite[Prop.~2.7]{gayral2004moyal} writes each coefficient as a finite sum of
\(c\star d\), and the constant commutators with the affine multipliers \(b_j\)
give
\[
L_{c\star d}[D_{\mathrm{eff}},L_{b_j}]
=L_c[D_{\mathrm{eff}},L_{b_j\star d}]
 -L_{c\star b_j}[D_{\mathrm{eff}},L_d],
\]
whose right-hand side consists of represented differential one-forms, since
\(b_j\star d\) and \(c\star b_j\) are Schwartz.  All products in this
identity use the fixed \(\varrho\)-ordering, with spinor identity factors
suppressed.

\begin{theorem}[Bounded inner approximants and cutoff removal]
\label{thm:strong-resolvent}
For every \(R>0\), \(D_R^{\varrho,r,s}\) is self-adjoint on
\(\operatorname{Dom}D_{\mathrm{eff}}^{\varrho,r,s}\), Schwartz spinors form
a core, and
\[
\left(\mathcal A_{\vartheta_{\mathrm{eff}},\varrho},\mathcal H,
\pi^M_{2,\varrho},D_R^{\varrho,r,s},I\otimes\sigma_3\right)
\]
is an even locally compact spectral triple.  Moreover,
\begin{equation}\label{eq:core-convergence}
D_R^{\varrho,r,s}\Psi\longrightarrow D_\infty\Psi,
\qquad
\Psi\in\mathcal S(\mathbb R^2)\otimes\mathbb C^2,
\end{equation}
and, for every \(z\in\mathbb C\setminus\mathbb R\),
\begin{equation}\label{eq:strong-resolvent}
\mathrm{s}\!-\!\lim_{R\to\infty}(D_R^{\varrho,r,s}-z)^{-1}
=(D_\infty-z)^{-1}.
\end{equation}
Finally,
\begin{equation}\label{eq:all-KK-equal}
[D_R^{\varrho,r,s}]=[D_\infty]=[D_{\mathrm{eff}}^{\varrho,r,s}]
\quad\text{in }KK^0(A_{\vartheta_{\mathrm{eff}},\varrho},\mathbb C).
\end{equation}
\end{theorem}

\begin{proof}
Writing \(D_R=D_{\mathrm{eff}}+B_R^{(\varrho)}\), one has
\[
B_R^{(\varrho)}=-\frac e{\sqrt2c_0}
\left(L_{A_{x,\varrho}^{(R)}}^{(\varrho)}\otimes\sigma_1
-\varepsilon_0L_{A_{y,\varrho}^{(R)}}^{(\varrho)}\otimes\sigma_2\right),
\]
which is bounded, self-adjoint and odd.  Hence the domain and core are
unchanged.  Algebra-domain invariance follows from the base triple, and
\[
[D_R,\pi^M_{2,\varrho}(a_\varrho)]
=[D_{\mathrm{eff}},\pi^M_{2,\varrho}(a_\varrho)]
 +[B_R^{(\varrho)},\pi^M_{2,\varrho}(a_\varrho)]
\]
is bounded.  The resolvent identity gives
\[
\pi^M_{2,\varrho}(a_\varrho)(D_R-i)^{-1}
=\pi^M_{2,\varrho}(a_\varrho)(D_{\mathrm{eff}}-i)^{-1}
 \bigl(I-B_R^{(\varrho)}(D_R-i)^{-1}\bigr),
\]
so local compactness follows from Theorem~\ref{thm:moyal-base-triple}.
Oddness of the perturbation preserves the grading.

It remains to identify the limit.  For an affine function \(p\), put
\(m_R=(u_R-1)p\).  For \(\varphi\in\mathcal S(\mathbb R^2)\), use the
normalized Fourier convention of
Section~\ref{sec:affine-left-right-general},
\[
\widehat\varphi(\xi)
=\frac1{(2\pi)^2}\int_{\mathbb R^2}e^{-iz\cdot\xi}\varphi(z)\,dz.
\]
The Weyl-ordered product then has the formula
\[
(m_R\star_{\vartheta_{\mathrm{eff}},1/2}\varphi)(z)
=\int_{\mathbb R^2}e^{iz\cdot\xi}
 m_R\!\left(z-\tfrac12\vartheta_{\mathrm{eff}}J_2\xi\right)
 \widehat\varphi(\xi)\,d\xi.
\]
Leibniz' rule and the annular support of the derivatives of \(u_R\) give
\[
|\partial^\alpha m_R(z)|\le C_\alpha\langle z\rangle,
\qquad R\ge1,
\]
for every multi-index \(\alpha\), with constants independent of \(R\).
For fixed \(z\), the integrand tends to zero and is bounded uniformly
in \(R\) by an integrable function of \(\xi\), since
\(\widehat\varphi\) is Schwartz.  Integration by parts with
\((1-\Delta_\xi)^N\) also yields, for \(N\ge2\),
\[
\bigl|(m_R\star_{\vartheta_{\mathrm{eff}},1/2}\varphi)(z)\bigr|
\le C_{N,\varphi}\langle z\rangle^{1-2N}.
\]
Pointwise convergence together with this \(L^2\)-majorant proves
\(m_R\star_{\vartheta_{\mathrm{eff}},1/2}\varphi\to0\) in \(L^2\).
Since \(T_\varrho A_j=A_j\), for \(\psi\in\mathcal S(\mathbb R^2)\)
one has
\[
T_\varrho\left(
 L_{A_{j,\varrho}^{(R)}}^{(\varrho)}-L_{A_j}^{(\varrho)}
\right)\psi
=\bigl((u_R-1)A_j\bigr)
 \star_{\vartheta_{\mathrm{eff}},1/2}(T_\varrho\psi).
\]
Unitarity of \(T_\varrho\) therefore gives
\(L_{A_{j,\varrho}^{(R)}}^{(\varrho)}\psi
 \to L_{A_j}^{(\varrho)}\psi\), which proves
\eqref{eq:core-convergence}.

For \(\Psi\) in the common Schwartz core and
\(\Xi=(D_\infty-z)\Psi\),
\[
(D_R-z)^{-1}\Xi-(D_\infty-z)^{-1}\Xi
=(D_R-z)^{-1}(D_\infty-D_R)\Psi\longrightarrow0.
\]
Since \(\|(D_R-z)^{-1}\|\le|\operatorname{Im}z|^{-1}\) and these \(\Xi\)
form a dense subspace, because Schwartz spinors are a core for
\(D_\infty\), the convergence extends to all of \(\mathcal H\), proving
\eqref{eq:strong-resolvent}.  Finally, for each fixed \(R>0\),
\(D_{\mathrm{eff}}+tB_R^{(\varrho)}\), \(0\le t\le1\), is a bounded
self-adjoint odd perturbation homotopy and gives
\([D_R]=[D_{\mathrm{eff}}]\).  The equality
\([D_\infty]=[D_{\mathrm{eff}}]\) follows from
Corollaries~\ref{thm:gauge-index} and~\ref{cor:full-affine-potentials}.
This proves \eqref{eq:all-KK-equal}.
\end{proof}

\section{Conclusion}\label{sec:conclusion}

The paper establishes a general left--right mechanism behind locally compact
spectral geometry on the two-dimensional Moyal plane.  For affine Dirac data
\(\eta=T_L\nu+T_Rw+c\),
\[
\det A_\eta=\vartheta^2\det T_R,
\qquad
\Gamma=\vartheta(\det T_L-\det T_R),
\]
and local compactness is equivalent to \(\det T_R\neq0\).  In the locally
compact regime, for fixed \(\vartheta\), \(|\det T_R|\) fixes the leading localized heat coefficient and Dixmier-trace volume, \(\operatorname{sgn}(\det T_R)\) fixes the oriented Gaussian \(K\)-homology pairing, and \(\Gamma\) distinguishes the Landau and continuous spectral regimes.  Wallet's Landau--Moyal family makes this division concrete: metric degeneracy occurs at its purely right-regular point, whereas local compactness fails at its purely left-regular point.

For NCQM, the labelled central cocycles of a fixed nondegenerate irreducible
sector determine the splitting \(V=V_\Pi\oplus V_X\), and the momentum
calculus selects the effective Moyal scale
\[
\vartheta_{\mathrm{eff}}
=\frac{\vartheta_0}{1-\vartheta_0B_0/\hbar_0}.
\]
One metaplectic transport carries the sector algebraic frame and Dirac frame to
the left--right Moyal realization.  The sector fixes \(T_R\), whereas affine
\(U(1)_\star\) background coupling changes only \(T_L\).  Hence the
transversality determinant, local spectral volume and even \(K\)-homology
class remain constant throughout the affine gauge family, while
\[
\Gamma(F_{xy})
=\frac12\left(1-\frac{eF_{xy}}{B_0}\right)
\]
drives the spectral transition.  At \(eF_{xy}=B_0\), the Landau levels give
way to continuous spectrum \(\mathbb R\) with empty point spectrum without
loss of local compactness.  The exact localized heat trace remains continuous
through this cancellation and yields the field-independent volume
\[
\operatorname{Tr}_\omega\!\left(
\pi^M_{2,\varrho}(a_\varrho)(1+D_\infty^2)^{-1}\right)
=2\left|\frac{\vartheta_0B_0}{\hbar_0}\right|\tau_{\vartheta_{\mathrm{eff}}}(a_\varrho).
\]

The extension to every real affine potential in
Corollary~\ref{cor:full-affine-potentials} preserves this mechanism.
Corollary~\ref{cor:metric-vs-spectral-degeneracy} also separates loss of
rank of the metric coefficients from spectral cancellation.  Thus the
affine case distinguishes three structures that need not degenerate together: the represented metric geometry, algebra--Dirac transversality, and
the fine spectrum.  Extending this separation to genuinely non-affine
unbounded Moyal gauge potentials, where self-adjointness and global symbol
control are no longer automatic, is a natural next problem.

A second direction suggested by Theorem~\ref{thm:affine-left-right} is a
$2n$-dimensional Moyal-plane analogue.  For affine data
$\eta=T_L\nu+T_Rw+c$ with nondegenerate deformation matrix $\Theta$, the
change of variables $(\nu,w)\mapsto(\nu,\eta)$ suggests that
$T_R\in GL_{2n}(\R)$ should remain the transversality condition, with local
spectral dimension $2n$ and volume coefficient proportional to
$|\det T_R|^{-1}$.  The scalar $\Gamma$ would be replaced by the
antisymmetric commutator matrix
\[
C=T_L\Theta T_L^{\mathsf T}-T_R\Theta T_R^{\mathsf T},
\]
whose rank and skew spectrum should control the higher-dimensional spectral
regime.  These expectations require a separate matrix-valued Clifford analysis
and are left for future work; the NCQM application here remains intrinsically
two-dimensional.

\section{Appendices}
\label{sec:appendices}

We give the compression argument used in the index pairing and collect
the phase and cocycle calculations underlying the sector construction.
Throughout the sector calculations, the admissibility assumptions of
Section~\ref{sec:cstar-gnc} remain in force and
\(\tau=\tau_{\hbar_0,\vartheta_0,B_0}\).

\subsection{Compression and the nonunital index pairing}
\label{app:compression-bridge}

\begin{lemma}[Compression bridge]
\label{lem:compression-bridge}
Let \((\mathcal A,\mathcal H,D,\chi)\) be an even locally compact
spectral triple, and let \(P=\pi(p)\) for a projection
\(p=p^*=p^2\in\mathcal A\).  Assume
\(P\operatorname{Dom}D\subset\operatorname{Dom}D\) and that
\([D,P]\) extends boundedly.  On \(\operatorname{Dom}D\), set
\[
D^{\mathrm{bd}}=PDP+(1-P)D(1-P),\qquad
D_P=D^{\mathrm{bd}}|_{P\mathcal H},\qquad
\operatorname{Dom}D_P=P\operatorname{Dom}D.
\]
Then \(D_P\) is self-adjoint, odd, and has compact resolvent.  With
\(F_T:=T(1+T^2)^{-1/2}\), one has
\begin{equation}\label{eq:compression-bridge}
PF_DP-F_{D_P}\in\mathcal K(P\mathcal H).
\end{equation}
Writing the positive chiral part as a map from
\(P\mathcal H^+\) to \(P\mathcal H^-\), the index pairing is
\begin{equation}\label{eq:compression-index}
\langle[p],[D]\rangle
=\operatorname{Index}(PF_D^+P)
=\operatorname{Index}(D_P^+).
\end{equation}
Here \(D_P^+\) carries its graph domain.  The same assertions hold
after finite matrix amplification.
\end{lemma}

\begin{proof}
The off-diagonal part
\[
B:=D-D^{\mathrm{bd}}
=PD(1-P)+(1-P)DP=[P,[P,D]]
\]
extends to a bounded self-adjoint operator.  Hence
\(D^{\mathrm{bd}}=D-B\) is self-adjoint on
\(\operatorname{Dom}D\), and \(P\) reduces it.  Its restriction
\(D_P\) is therefore self-adjoint on the stated domain.  Since
\([P,\chi]=0\), both operators are odd.

Local compactness gives \(P(D-z)^{-1}\in\mathcal K\) for every
\(z\notin\mathbb R\).  The resolvent identity yields, on
\(P\mathcal H\),
\[
\begin{aligned}
(D_P-z)^{-1}
&=P(D^{\mathrm{bd}}-z)^{-1}P\\
&=P(D-z)^{-1}P
 +P(D-z)^{-1}B(D^{\mathrm{bd}}-z)^{-1}P.
\end{aligned}
\]
Both terms are compact, so \(D_P\) has compact resolvent.

Put \(\mu=\sqrt{1+\lambda^2}\).  For self-adjoint \(T\),
\[
F_T=\frac2\pi\int_0^\infty T(T^2+\mu^2)^{-1}\,d\lambda,
\]
where the improper integral converges strongly.  Subtract the truncated
formulas for \(D\) and \(D^{\mathrm{bd}}\), and apply the
resolvent identity to obtain
\begin{equation}\label{eq:bounded-transform-difference}
F_D-F_{D^{\mathrm{bd}}}
=-\frac1\pi\int_0^\infty\sum_{\pm}
(D\pm i\mu)^{-1}B(D^{\mathrm{bd}}\pm i\mu)^{-1}\,d\lambda.
\end{equation}
This difference integral converges in norm: the sum is bounded by
\(2\|B\|/(1+\lambda^2)\).  After compression by \(P\), each
integrand is compact because \(P(D\pm i\mu)^{-1}\) is compact.
As \(PF_{D^{\mathrm{bd}}}P=F_{D_P}\), this proves
\eqref{eq:compression-bridge}.

Compact resolvent gives \(F_{D_P}^2-I\in\mathcal K(P\mathcal H)\),
so \(F_{D_P}^+\) is Fredholm.  Its kernel and cokernel agree with
those of \(D_P^+\).  Compact perturbation and the even
Fredholm-module definition of the pairing now give
\eqref{eq:compression-index}; see
\cite{CareyGayralRennieSukochev2014} for the nonunital index-pairing
framework.  No summability assumption is needed for this argument.
\end{proof}

\subsection{Weyl phases, symmetrization, and covariance}
\label{App:weyl}

Write the actions in \eqref{eq:unitary_actions} as
\[
(U_j(t)\phi)(x)=e^{it b_j\cdot x}\phi(x+t a_j),
\qquad x\in\mathbb R^2,
\]
where \(a_j,b_j\in\mathbb R^2\) are their translation and phase
coefficients.  Comparing the two compositions gives the single formula
\[
U_j(t)U_k(u)
=e^{itu(a_j\cdot b_k-a_k\cdot b_j)}U_k(u)U_j(t).
\]
For example, the coefficient for the pair \((j,k)=(1,3)\) is
\[
\frac{\vartheta_0B_0(r+s-rs)-\hbar_0
      -s\vartheta_0B_0(1-r)}
     {\hbar_0(r\vartheta_0B_0-\hbar_0)}
=\frac1{\hbar_0}.
\]
For all pairs, the coefficient matrix is
\[
\bigl(a_j\cdot b_k-a_k\cdot b_j\bigr)_{j,k}
=\frac1{\hbar_0}
\begin{pmatrix}
0&-B_0&1&0\\
B_0&0&0&1\\
-1&0&0&-\vartheta_0/\hbar_0\\
0&-1&\vartheta_0/\hbar_0&0
\end{pmatrix}
=-2\pi\tau.
\]
This verifies all six relations in
\eqref{eq:all_weyl_commutation_relation} and the cancellation of the
presentation parameters \((r,s)\).

Moving each factor of \(U(p)\) left past the factors of \(U(q)\) with larger
index contributes \(e^{2\pi i p_j\tau_{jk}q_k}\) for \(j<k\), which gives
\eqref{eq:U-product}. With \(s_\tau(q)=\pi\sum_{j<k}q_j\tau_{jk}q_k\),
\[
s_\tau(q)+s_\tau(p)-s_\tau(q+p)+2\pi\sum_{j<k}p_j\tau_{jk}q_k
=-\pi q^{\mathsf T}\tau p,
\]
so \(W=e^{is_\tau}U\) satisfies \eqref{eq:sector-weyl-multiplier};
taking \(p=-q\) in \eqref{eq:U-product} gives \(W(q)^*=W(-q)\).

\paragraph{The covariance phase.}
Set \(k'=(-\hbar_0k_1,-\hbar_0k_2,0,0)\).  The exponential
in \eqref{eq:projective_representation} gives
\[
U_k=e^{-i(k_1\Pi_x+k_2\Pi_y)}=W(k').
\]
The multiplier law and the displayed matrix of \(\tau\) give
\[
W(v)W(q)W(v)^*=e^{-2\pi i v^{\mathsf T}\tau q}W(q),\qquad
-2\pi k'^{\mathsf T}\tau q=k\cdot q^\sharp,
\]
where \(q^\sharp=(B_0q_2-q_3,-B_0q_1-q_4)\).  Thus
\begin{equation}\label{eq:app-cov-Weyl-cov}
U_kW(q)U_k^*=e^{ik\cdot q^\sharp}W(q).
\end{equation}
This is the explicit phase calculation used in
Proposition~\ref{prop:sector-covariance-derivations}.

\subsection{Cocycle and central-character normalization}
\label{App:matrix}

Use the local exponents \(\xi,\xi',\xi''\) defined in
Section~\ref{sec:cstar-gnc}, in the parameter ordering
\((\Pi_x,\Pi_y,X,Y)\).  Directly from
\eqref{eq:skewsymmatrix},
\begin{equation}\label{eq:appendix-tau-expanded}
\begin{aligned}
p^{\mathsf T}\tau q
&=\frac{q_1p_3+q_2p_4-q_3p_1-q_4p_2}{2\pi\hbar_0}\\
&\quad-\frac{\vartheta_0(q_3p_4-q_4p_3)}{2\pi\hbar_0^2}
       -\frac{B_0(q_1p_2-q_2p_1)}{2\pi\hbar_0}.
\end{aligned}
\end{equation}
Equivalently,
\begin{equation}\label{eq:appendix-local-exponents}
\pi p^{\mathsf T}\tau q
=\frac1{\hbar_0}\xi(q,p)
 -\frac{\vartheta_0}{\hbar_0^2}\xi'(q,p)
 -\frac{B_0}{\hbar_0}\xi''(q,p).
\end{equation}
Since \(p^{\mathsf T}\tau q=-q^{\mathsf T}\tau p\), the
multiplier \(e^{-\pi i q^{\mathsf T}\tau p}\) is precisely the
central-character multiplier for
\[
\kappa_0=\left(\hbar_0^{-1},
-\vartheta_0\hbar_0^{-2},-B_0\hbar_0^{-1}\right).
\]

For the general central character and form \(\beta_\kappa\) in
Proposition~\ref{prop:cstar-bundle}, the matrix defined by
\(\beta_\kappa(q,p)=q^{\mathsf T}M_\kappa p\) is
\[
M_\kappa=\frac12
\begin{pmatrix}
0&\kappa_B&\kappa_\hbar&0\\
-\kappa_B&0&0&\kappa_\hbar\\
-\kappa_\hbar&0&0&\kappa_\vartheta\\
0&-\kappa_\hbar&-\kappa_\vartheta&0
\end{pmatrix},\qquad
\operatorname{Pf}M_\kappa
=\frac{\kappa_\vartheta\kappa_B-\kappa_\hbar^2}{4}.
\]
This fixes the normalization in
\eqref{eq:kappa-nondegeneracy} exactly.
At the physical character, \(M_{\kappa_0}=-\pi\tau\), in
agreement with \eqref{eq:central-character-physical-chart} and
\eqref{eq:sector-tau-determinant}.

\section*{Acknowledgements}
M. R. J. would like to thank Onirban Islam for insightful discussions and gratefully acknowledges his help with references.
T. K. S. would like to thank his family and friends for their love, support, and strength.

\section*{Declarations}

\paragraph{Funding.}
Part of this work was carried out with support from TWAS, Italy, under research
grant no.~22-213 RG/MATHS/AS\_I-FR 3240324946. The grant concluded in 2024.

\paragraph{Use of artificial intelligence tools.}
During the preparation of this manuscript, the authors used OpenAI's ChatGPT and Anthropic's Claude to assist with language editing, reference formatting, and the preliminary identification of notational inconsistencies. All final proofs, results, and formulations were independently derived and rigorously verified by the authors, who take full responsibility for the content and integrity of the manuscript.

\end{document}